\documentclass[11pt]{article}
\usepackage[margin=1in]{geometry}
\usepackage[latin2]{inputenc}
\usepackage[english]{babel}
\usepackage{amsmath}
\usepackage[all=normal,paragraphs=tight,bibliography=tight]{savetrees}
\usepackage{todonotes}
\usepackage{amsthm}
\usepackage{amssymb}
\usepackage{microtype}
\usepackage{xspace}
\usepackage{comment}
\usepackage{letltxmacro}
\usepackage{comment}
\usepackage{enumerate}
\usepackage{tikz}
\usepackage{url}
\usepackage{eufrak}

\usepackage{graphicx}
\usepackage{caption}
\usepackage{subcaption}

\newtheorem{theorem}{Theorem}[section]
\newtheorem{lemma}[theorem]{Lemma}
\newtheorem{claim}[theorem]{Claim}

\theoremstyle{definition}
\newtheorem{definition}[theorem]{Definition}

\newcommand{\conn}{\mu}
\newcommand{\tw}{\mathbf{tw}}

\newcommand{\parent}{\mathrm{parent}}

\newcommand{\Oh}{\ensuremath{\mathcal{O}}}

\newcommand{\GI}{{\sc{Graph Isomorphism}}\xspace}

\def\cqedsymbol{\ifmmode$\lrcorner$\else{\unskip\nobreak\hfil
\penalty50\hskip1em\null\nobreak\hfil$\lrcorner$
\parfillskip=0pt\finalhyphendemerits=0\endgraf}\fi} 

\newcommand{\cqed}{\renewcommand{\qed}{\cqedsymbol}}

\newcommand{\executeiffilenewer}[3]{%
\ifnum\pdfstrcmp{\pdffilemoddate{#1}}%
{\pdffilemoddate{#2}}>0%
{\immediate\write18{#3}}\fi%
} 
\newcommand{\app}{$\spadesuit$}
\newcommand{\imp}[2]{#1^{\langle#2\rangle}}

\newcommand{\termfam}{\mathbb{T}}
\newcommand{\Op}{\mathbb{O}}
\newcommand{\term}{\mathbf{t}}
\newcommand{\leaf}{\mathfrak{l}}
\newcommand{\join}{\mathfrak{j}}
\newcommand{\introduce}{\mathfrak{i}}
\newcommand{\forget}{\mathfrak{f}}
\newcommand{\edge}{\mathfrak{e}}
\newcommand{\Labels}{\Sigma}
\newcommand{\used}{\mathbf{used}}
\newcommand{\baggraph}{\mathbf{bag}}
\newcommand{\labG}{\mathfrak{G}}
\newcommand{\lab}{\lambda}
\newcommand{\bagfam}{\mathcal{B}}
\newcommand{\canon}{\mathfrak{c}}

\newcommand{\tl}{\vartriangleleft}
\newcommand{\tg}{\vartriangleright}
\newcommand{\tleq}{\trianglelefteq}

\newcommand{\pot}{\Phi}

\title{Fixed-parameter tractable canonization and isomorphism test\\ for graphs of bounded treewidth%
\thanks{%
A preliminary version of this paper has been presented at FOCS 2014.
D. Lokshtanov is supported by the BeHard grant under the recruitment programme of the of Bergen Research Foundation.
The research of M. Pilipczuk and M. Pilipczuk leading to these results has received funding from the European Research Council under the European Union's Seventh Framework Programme (FP/2007-2013) / ERC Grant Agreement n. 267959.
S. Saurabh is supported by PARAPPROX, ERC starting grant no. 306992.}}
\author{
  Daniel Lokshtanov
  \thanks{
    Department of Informatics, University of Bergen, Norway, \texttt{daniello@ii.uib.no}.
  }
  \and
  Marcin Pilipczuk
  \thanks{
     Department of Computer Science, University of Warwick, UK, \texttt{M.Pilipczuk@dcs.warwick.ac.uk}.
  }
  \and
  Micha\l{} Pilipczuk
  \thanks{
    Institute of Informatics, University of Warsaw, Poland, \texttt{michal.pilipczuk@mimuw.edu.pl}.
  }
  \and 
  Saket Saurabh\thanks{
    Institute of Mathematical Sciences, India, \texttt{saket@imsc.res.in}, and
    Department of Informatics, University of Bergen, Norway, \texttt{Saket.Saurabh@ii.uib.no}.
  }
}

\date{}

\begin{document}

\begin{titlepage}
\def\thepage{}
\thispagestyle{empty}
\maketitle

\begin{abstract}
We give a fixed-parameter tractable algorithm that, given a parameter $k$ and two graphs $G_1,G_2$, either concludes that one of these graphs has treewidth at least $k$, or determines whether $G_1$ and $G_2$ are isomorphic. The running time of the algorithm on an $n$-vertex graph is $2^{\Oh(k^5\log k)}\cdot n^5$, and this is the first fixed-parameter algorithm for \GI parameterized by treewidth. 

Our algorithm in fact solves the more general {\em{canonization}} problem. We namely design a procedure working in $2^{\Oh(k^5\log k)}\cdot n^5$ time that, for a given graph $G$ on $n$ vertices, either concludes that the treewidth of $G$ is at least $k$, or:
\begin{itemize}
\item finds in an isomorphic-invariant way a graph $\canon(G)$ that is isomorphic to $G$;
\item finds an isomorphism-invariant {\em{construction term}} --- an algebraic expression that encodes $G$ together with 
a tree decomposition of $G$ of width $\Oh(k^4)$.
\end{itemize}
Hence, the isomorphism test reduces to verifying whether the computed isomorphic copies or the construction terms
for $G_1$ and $G_2$ are equal.

\end{abstract}
\end{titlepage}

\section{Introduction}\label{sec:intro}
\GI is one of the most fundamental graph problems: given two graphs $G_1,G_2$, decide whether $G_1$ and $G_2$ are {\em{isomorphic}}, i.e., there exists a bijection $\phi$ between $V(G_1)$ and $V(G_2)$ such that $uv\in E(G_1)$ if and only if $\phi(u)\phi(v)\in E(G_2)$. Despite extensive research on the topic, it is still unknown whether the problem can be solved in polynomial time. On the other hand, there are good reasons to believe that \GI is not NP-hard either, since NP-hardness of \GI would imply a collapse of the polynomial hierarchy~\cite{Schoning88}.

A significant amount of effort has been put into understanding and broadening the spectrum of classes of graphs where polynomial-time isomorphism tests can be designed. Perhaps the most important example is the classic algorithm of Babai and Luks~\cite{BabaiL83,Luks82}, which solves \GI on graphs of maximum degree $d$ in time $n^{\Oh(d)}$. On the other hand, following polynomial-time solvability of \GI on planar graphs~\cite{HopcroftT72,HopcroftT73,HopcroftW74,Weinberg} it has been investigated how more general topological constraints can be exploited to design efficient algorithms for the problem. Isomorphism tests for graphs of genus $g$ working in time $n^{\Oh(g)}$ were proposed independently by Filotti and Mayer~\cite{FilottiM80} and by Miller~\cite{Miller80}. These results were improved by Ponomarenko~\cite{ponomarenko}, who gave an $\Oh(n^{f(|H|)})$ algorithm for graphs excluding a fixed graph $H$ as a minor, for some function $f$. The result of Ponomarenko implies also that \GI can be solved in polynomial time on graphs of constant treewidth. A simple algorithm for graphs of treewidth $k$ running in time $\Oh(n^{k+4.5})$ was independently given by Bodlaender~\cite{Bodlaender90}. Finally, we  mention the result of Arnborg and Proskurowski~\cite{ArnborgP92}, who gave canonical representation of partial $2$- and $3$-trees.


Recently, Grohe and Marx~\cite{marx-grohe-arxiv,marx-grohe} obtained a structure theorem for graphs excluding a fixed graph $H$ as a topological minor. This theorem roughly states that such graphs can be decomposed along small separators into parts that are either $H$-minor-free, or of almost bounded degree (in terms of $|H|$). Using previous algorithms for $H$-minor-free graphs~\cite{ponomarenko} and bounded-degree graphs~\cite{BabaiL83,Luks82}, they managed to show that \GI can be solved in $\Oh(n^{f(|H|)})$ time for $H$-topological-minor-free graphs, for some function $f$. Note that this result generalizes both
the algorithms for \GI on minor free-graphs~\cite{ponomarenko} and on bounded degree graphs~\cite{BabaiL83,Luks82}. The work of Grohe and Marx constitutes the current frontier of polynomial-time solvability of \GI.

Observe that in all the aforementioned results the exponent of the polynomial depends on the considered parameter, be it the maximum degree, genus, treewidth, or the size of the excluded (topological) minor. In the field of parameterized complexity such algorithms are called XP algorithms (for {\em{slice-wise polynomial}}), and are often compared to the more efficient FPT algorithms (for {\em{fixed-parameter tractable}}), where the running time is required to be of the form $f(k)\cdot n^c$. Here $k$ is the parameter, $f$ is a computable function, and $c$ is a universal constant independent of $k$. One of the main research directions in parameterized complexity is to consider problems that admit XP algorithms and determine whether they admit an FPT algorithm;
we refer to the monographs~\cite{DowneyF13,FlumGroheBook} for more information on parameterized complexity. 
It is therefore natural to ask which of the the aforementioned results on \GI can be improved to fixed-parameter tractable algorithms.

Prior to this work very little was known about such improvements. In particular, the existence of FPT algorithms for \GI parameterized by maximum degree, genus or treewidth of the input graph have remained intriguing open problems.
%
%
%
%
%
A reader familiar with the algorithmic aspects of treewidth might find it surprising that the existence of an FPT algorithm for \GI parameterized by treewidth is a difficult problem. The parameter has been studied intensively during the last 25 years, and is quite well-understood. 
Many problems that are very hard on general graphs become polynomial time, or even linear-time solvable on graphs of constant treewidth. For the vast majority of these problems, designing an FPT algorithm parameterized by treewidth boils down to designing a straightforward dynamic programming algorithm over the decomposition. For \GI this is not the case, even the relatively simple  $\Oh(n^{k+4.5})$ time algorithm of Bodlaender~\cite{Bodlaender90} is structurally quite different from most algorithms on bounded treewidth graphs. This might be the reason why \GI was one of very few remaining problems of fundamental nature, whose fixed-parameter tractability when parameterized by treewidth was unresolved.
%
%
%
%

Therefore, fixed-parameter tractability of \GI parameterized by treewidth has been considered an important open problem in parameterized complexity for years. This question (and its weaker variants for width measures lower-bounded by treewidth) was asked explicitly in~\cite{open-iwpec08,BoulandDK12,marx-grohe-arxiv,KawarabayashiM08,KratschS10,Otachi12,YamazakiBFT99}, and appears on the open problem list of the recent edition of the monograph of Downey and Fellows~\cite{DowneyF13}. 

Most of the related work on the parameterized complexity of \GI with respect to width measures considers parameters that are always at least as large as treewidth. The hope has been that insights gained from these considerations might eventually lead to settling the main question. In particular, fixed-parameter tractable algorithms for \GI has been given for the following parameters: tree-depth~\cite{BoulandDK12}, feedback vertex set number~\cite{KratschS10}, connected path distance width~\cite{Otachi12}, and rooted tree distance width~\cite{YamazakiBFT99}. Very recent advances by Otachi and Schweitzer~\cite{pascale} give FPT algorithms for parameterizations by root-connected tree distance width and by connected strong treewidth. Even though all these parameters are typically much larger than treewidth, already settling fixed-parameter tractability for them required many new ideas and considerable technical effort. This supports the statement by Kawarabayashi and Mohar~\cite{KawarabayashiM08} that ``[...] even for graphs of bounded treewidth, the graph isomorphism problem is not trivial at all''.

\paragraph*{Our results.} In this paper we answer the question of fixed-parameter tractability of \GI parameterized by treewidth in the affirmative, by proving the following theorem:

\begin{theorem}\label{main:isomorphism}
There exists an algorithm that, given an integer $k$ and two graphs $G_1,G_2$ on $n$ vertices, works in time $2^{\Oh(k^5\log k)}\cdot n^5$ and either correctly concludes that $\tw(G_1)\geq k$ or $\tw(G_2)\geq k$, or determines whether $G_1$ and $G_2$ are isomorphic.
\end{theorem}

In fact, we prove a stronger statement, that is, we provide a {\em{canonization}} algorithm for graphs of bounded treewidth.
This can be defined in two manners. 
First, following the formalism of Grohe and Marx~\cite{marx-grohe-arxiv,marx-grohe}, we show an algorithm that, given $G$,
constructs a canonical graph $\canon(G)$ on the vertex set $\{1,2,\ldots,|V(G)|\}$ isomorphic to $G$, such that $\canon(G_1)=\canon(G_2)$ whenever $G_1$ and $G_2$ are isomorphic.

\begin{theorem}\label{simple:canonization}
There exists an algorithm that, given an integer $k$ and a graph $G$ on $n$ vertices, works in time $2^{\Oh(k^5\log k)}\cdot n^5$ and either correctly concludes that $\tw(G)\geq k$ or outputs, in an isomorphic-invariant way, a graph $\canon(G)$ that is isomorphic to $G$,
together with a mapping $\phi:V(G) \to V(\canon(G))$ that certifies this isomorphism.
\end{theorem} 

A second way is through so-called {\em{construction terms}}, defined formally in Section~\ref{sec:canon}: they are algebraic expressions encoding construction procedures for graphs of bounded treewidth. Thus, construction terms can be seen as an alternative definition of treewidth via graph grammars (see Lemma~\ref{lem:tw-ct}).

\begin{theorem}\label{main:canonization}
There exists an algorithm that, given a graph $G$ and a positive integer $k$, in time $2^{\Oh(k^5\log k)}\cdot n^{5}$ either correctly concludes that $\tw(G)\geq k$, or outputs an isomorphism-invariant term $\term$ that constructs $G$ and uses at most $\Oh(k^4)$ labels. Moreover, this term has length at most $\Oh(k^4\cdot n)$.
\end{theorem} 

The approach to treewidth and tree decompositions via graph grammars and tree automata dates back to the earliest works on this subject, and is the foundation of intensive research on links between treewidth and monadic second-order logic; we refer to a recent monograph of Courcelle and Engelfriet~\cite{0030804} for a more thorough introduction. Unfortunately there is currently no agreed standard notation for these concepts. In order to ensure clarity of notation, we introduce our own formalism in Section~\ref{sec:terms}. 

Let us point out that for all the aforementioned classes where \GI can be solved in XP time, corresponding XP canonization algorithms were also developed: for bounded-degree graphs by Babai and Luks~\cite{BabaiL83}, for $H$-minor-free graphs by Ponomarenko~\cite{ponomarenko}, and for $H$-topological-minor-free graphs by Grohe and Marx~\cite{marx-grohe-arxiv,marx-grohe}.

We also remark that Theorems~\ref{main:isomorphism} and~\ref{simple:canonization} are straightforward corollaries of Theorem~\ref{main:canonization}. For Theorem~\ref{main:isomorphism}, we apply the algorithm of Theorem~\ref{main:canonization} to both $G_1$ and $G_2$, and verify whether the obtained terms are equal. Similarly, for Theorem~\ref{simple:canonization}, we may order the vertices of the input graph $G$
according to the order of they appearance in the canonical term. Formal proofs are provided in Section~\ref{sec:canon-cors}.

\paragraph*{Our techniques.}
We now sketch the main ideas behind the proof of Theorem~\ref{main:canonization}. The starting point is the classic algorithm of Bodlaender~\cite{Bodlaender90} that resolves isomorphism of graphs of treewidth $k$ in time $\Oh(n^{k+4.5})$. Essentially, this algorithm considers all the $(k+1)$-tuples of vertices of each of the graphs as potential bags of a tree decomposition, and tries to assemble both graphs from these building blocks in the same manner using dynamic programming.
It turns out that with a slight modification, the algorithm of Bodlaender can be extended to solve the canonization problem as well. We now direct our attention to speeding up the algorithm.

Our idea is the following: if we were able to constrain ourselves to a small enough family of potential bags, then basically the same algorithm restricted to the pruned family of states would work in FPT time. For this to work we need the family to be of size $f(k)\cdot n^c$ for some function $f$ and constant $c$. Furthermore, we would need an algorithm that given as input the graph $G$, computes this family in FPT time. Finally, we would need this family of bags to be {\em isomorphism invariant}. Informally, we want the pruned family of bags to only depend on the (unlabelled) graph $G$, and not on the labelling of vertices of $G$ given as input. A formal definition of what we mean by isomorphism invariance is given in the preliminaries.
%
%
%

Therefore, the goal is to find a family $\bagfam\subseteq 2^{V(G)}$ of potential bags that is on one hand isomorphism-invariant and reasonably small, and on the other hand it is rich enough to contain all the bags of some tree decomposition of width at most $g(k)$, for some function $g$. Coping with this task is the main contribution of this paper, and this is achieved in Theorems~\ref{thm:no-clique-seps},~\ref{thm:clique-seps}, and~\ref{thm:redsizes}. 
The idea that finding an isomorphism-invariant family of potential bags of size $f(k)\cdot n^c$ is sufficient for designing an isomorphism test was first observed by Otachi and Schweitzer~\cite{pascale}.%
\footnote{We remark that even though the work of Otachi and Schweitzer was announced almost simultaneously with our work, we believe it should be considered prior to our results. Their results were presented at a Shonan meeting as early as in May 2013~\cite{shonan}.}
However, their approach for proving this initial step is completely different.


The crucial idea of our construction of a small isomorphism-invariant family of bags 
is to start with the classic $4$-approximation algorithm for treewidth given in the Graph Minors series by Robertson and Seymour~\cite{gm13}; we also refer to the textbook of Kleinberg and Tardos~\cite{0015106} for a comprehensive exposition of this algorithm. Since a good understanding of this algorithm is necessary for our considerations, let us recall it briefly. 

The algorithm of Robertson and Seymour constructs a tree decomposition of the input graph $G$ in a top-down manner. More precisely, it is a recursive procedure that at each point maintains a separator $S$ of size at most $3k+\Oh(1)$, which separates the part of $G$ we are currently working with (call it $H$) from the rest. The output of the procedure is a tree decomposition of $H$ with the set $S$ as the top adhesion. At each recursive call the algorithm proceeds as follows. If $S$ is small, the algorithm adds to $S$ an arbitrarily chosen vertex $u\in V(H)$. If $S$ is large, the algorithm attempts  to break $S$ into two roughly equal-sized pieces using a separator $X$ of size at most $k+1$. The fact that the graph has treewidth at most $k$ ensures that such a separator $X$ will always exist. The new set $S'$, defined as $S\cup \{u\}$ or $S\cup X$, becomes the top-most bag. Below this bag we attach tree decompositions obtained from the recursive calls for instances $(H:=G[N[Z]],S:=N(Z))$, where $Z$ iterates over the family of vertex sets of the connected components of $H\setminus S'$. The crucial point is that if $S$ is large (of size roughly $3k$) and $X$ is of size at most $k+1$, then every such component $Z$ will neighbour at most $|S|$ vertices of $S'$, and hence the size of $S$ will not increase over the course of the recursion.

Our high-level plan is to modify this algorithm so that it works in an (almost) isomorphism-invariant manner, and then return all the produced bags as the family $\bagfam$. At a glance, it seems that the algorithm inherently uses two `very non-canonical' operations: adding an arbitrarily chosen vertex $u$ and breaking $S$ using an arbitrarily chosen separator $X$. 
%
Especially canonizing the choice of the separator $X$ seems like a hard nut to crack. We circumvent this obstacle in the following manner: we take $X$ to be the union of `all possible' separators that break $S$ evenly. The problem that now arises is that it is not clear how to bound the number of neighbours in $S'=S\cup X$ that can be seen from a connected component we want to recurse on; recall that we needed to bound this number by $|S|$, in order to avoid an explosion of the bag sizes throughout the course of the algorithm. The crucial technical insight of the paper, proved in Lemma~\ref{lem:magical}, is that if one defines `all possible separators' in a very careful manner, then this bound holds. The proof of Lemma~\ref{lem:magical} relies on a delicate argument that exploits submodularity of vertex separations. 


Even if the problem of canonizing the choice of $X$ is solved, we still need to cope with the arbitrary choice of $u$ in case $S$ is small. It turns out that the problem appears essentially only if the set $S$ is very well connected: for every two vertices $x,y\in S$, the vertex flow between $x$ and $y$ is more than $k$. In other words, the set $S$ constitutes a {\em{clique separator}} in the {\em{$k$-improved graph of $G$}}, i.e., a graph derived from $G$ by making every pair of vertices with vertex flow more than $k$ adjacent. It is known that for the sake of computing tree decompositions of width $k$ one can focus on the $k$-improved graph rather than the original one (see, e.g.,~\cite{Bodlaender96} and Lemma~\ref{lem:same-decomp}). Therefore, our algorithm will work without any problems provided that the $k$-improved graph of the input one does not admit any clique separators. The produced tree decomposition has width $2^{\Oh(k\log k)}$, due to a possibly exponential number of separators breaking $S$ at each step, and is isomorphism-invariant up to the choice of a single vertex from which the whole procedure begins. By running the algorithm from every possible starting point and computing the union of the families of bags of all the obtained decompositions, we obtain an isomorphism-invariant family of potential bags of size $\Oh(n^2)$. This result is obtained in Theorem~\ref{thm:no-clique-seps}, which summarizes the case when no clique separators are present.

However, the behaviour of clique separators in the graph has been well understood already in the 1980s, starting from the work of Tarjan~\cite{Tarjan85},
and studied intensively from a purely graph-theoretical viewpoint. 
It turns out that all inclusion-wise minimal clique separators of a graph form a tree-like structure, giving raise to so-called {\em{clique minimal separator decomposition}},
   which decomposes the graph into pieces that do not admit any clique separators. These pieces are often called {\em{atoms}}.
Most importantly for us, the set of atoms of a graph is isomorphism-invariant, and can be computed in polynomial time. Therefore, in the general case we can compute the clique minimal separator decomposition of the $k$-improved graph of the input graph, run the algorithm for the case of no clique separators on each atom separately, and finally output the union of all the obtained families. This result is obtained in Theorem~\ref{thm:clique-seps}. 
We refer to the introductory paper of Berry et al.~\cite{BerryPS10} for more information on clique separators.

The family $\bagfam$ given by Theorem~\ref{thm:clique-seps} is essentially already sufficient for running the modified algorithm of Bodlaender~\cite{Bodlaender90} on it, and thus resolving fixed-parameter tractability of \GI parameterized by treewidth. However, the bags contained in the family $\bagfam$ may be of size as much as $2^{\Oh(k\log k)}$. 
Our canonization algorithm considers 
every permutation of every candidate bag. Hence, this would result in a double-exponential dependence on $k$ in the running time. In Section~\ref{sec:redsizes} we demonstrate how to reduce this dependence to $2^{\text{poly}(k)}$. More precisely, we prove that instead of the original family $\bagfam$, we can consider a modified family $\bagfam'$ constructed as follows: for every $B\in \bagfam$, we replace $B$ with all the subsets of $B$ that have size $\Oh(k^4)$. Thus, every bag of $\bagfam$ gives raise to $\binom{2^{\Oh(k\log k)}}{\Oh(k^4)}=2^{\Oh(k^5\log k)}$ sets in $\bagfam'$, and hence $|\bagfam'|\leq 2^{\Oh(k^5\log k)}\cdot |\bagfam|$. However, now every bag of $\bagfam'$ has only $2^{\Oh(k^4\log k)}$ possible permutations, instead of a number that is double-exponential in $k$. In this manner, we can trade a possible explosion of the size of the constructed family for a polynomial upper bound on the cardinality of its members. This trade-off is achieved in Theorem~\ref{thm:redsizes}, and leads to a better time complexity of the canonization algorithm.

\paragraph*{Organization of the paper.} Section~\ref{sec:prelims} contains preliminaries. Sections~\ref{sec:no-cliqueseps},~\ref{sec:cliqueseps},~\ref{sec:redsizes} contain proofs of Theorems~\ref{thm:no-clique-seps},~\ref{thm:clique-seps},~\ref{thm:redsizes}, respectively. In Section~\ref{sec:canon} we utilize Theorem~\ref{thm:redsizes} to present the canonization algorithm, i.e., to prove Theorems~\ref{main:isomorphism} and~\ref{main:canonization}. In particular, Section~\ref{sec:terms} introduces the formalism of construction terms. In Section~\ref{sec:conc} we gather concluding remarks. Proofs of lemmas denoted by (\app) are straightforward, and have been moved to the appendix in order not to disturb the flow of the arguments.

\section{Preliminaries}\label{sec:prelims}
In most cases, we use standard graph notation, see e.g.~\cite{diestel}.

\paragraph*{Separations, separators, and clique separators.} We recall here standard definitions and facts about separations and separators in graphs.

\begin{definition}[\bf separation]
A pair $(A,B)$ where $A \cup B = V(G)$ is called a {\em separation}
if $E(A \setminus B, B \setminus A) = \emptyset$.
The {\em{separator}} of $(A,B)$ is $A\cap B$ and the {\em order} of a separation $(A,B)$ is $|A \cap B|$.
\end{definition}

Let $X,Y\subseteq V(G)$ be two not necessarily disjoint subsets of vertices. Then a separation $(A,B)$ is an {\em{$X-Y$ separation}} if $X\subseteq A$ and $Y\subseteq B$. The classic Menger's theorem states that for given $X,Y$, the minimum order of an $X-Y$ separation is equal to maximum vertex-disjoint flow between $X$ and $Y$ in $G$. This minimum order (denoted further $\conn(X,Y)$) can be computed in polynomial time, using for instance the Ford-Fulkerson algorithm. Moreover, among the $X-Y$ separations of minimum order there exists a unique one with inclusion-wise minimal $A$, and a unique one with inclusion-wise minimal $B$. We will call these minimum-order $X-Y$ separations {\em{pushed towards $X$}} and {\em{pushed towards $Y$}}, respectively. It is known that the Ford-Fulkerson algorithm can actually find the minimum order $X-Y$ separations that are pushed towards $X$ and $Y$ within the same running time.

For two vertices $x,y\in V(G)$, by $\conn(x,y)$ we denote the minimum order of a separation $(A,B)$ in $G$ such that $x\in A\setminus B$ and $y\in B\setminus A$. Note that if $xy\in E(G)$ then such a separation does not exist; in such a situation we put $\conn(x,y)=+\infty$. Again, classic Menger's theorem states that for nonadjacent $x$ and $y$, the value $\conn(x,y)$ is equal to the maximum number of internally vertex-disjoint $x-y$ paths that can be chosen in the graph, and this value can be computed in polynomial time using the Ford-Fulkerson algorithm. The notions of minimum-order separations pushed towards $x$ and $y$ are defined analogously as before.
If the graph we are referring to is not clear from the context, we put it in the subscript by the symbol $\conn$.

We emphasize here that, contrary to $X-Y$ separations, in an $x-y$ separation we require $x \in A \setminus B$ and $y \in B \setminus A$, that is, the separator $A \cap B$
cannot contain $x$ nor $y$. This, in particular, applies to minimum-order $x-y$ separations pushed towards $x$ or $y$.

\begin{definition}[\bf clique separation]
A separation $(A,B)$ is called a {\em{clique separation}} if $A\setminus B\neq \emptyset$, $B\setminus A\neq \emptyset$, and $A\cap B$ is a clique in $G$.
\end{definition}

Note that an empty set of vertices is also a clique, hence any separation with an empty separator is in particular a clique separation. We will say that a  graph is {\em{clique separator free}} if it does not admit any clique separation. Such graphs are often called also {\em{atoms}}, see e.g.~\cite{BerryPS10}. From the previous remark it follows that every clique separator free graph is connected.

\paragraph*{Tree decompositions.} In this paper it is most convenient to view tree decompositions of graphs as rooted. The following notation originates in Marx and Grohe~\cite{marx-grohe}, and we use an extended version borrowed from~\cite{CyganLPPS13}.

Let $T$ be a rooted tree and let $t$ be any non-root node of $T$. The parent of $t$ in $T$ will be denoted by $\parent(t)$. A node $s$ is a {\emph{descendant}} of $t$, denoted $s\preceq t$, if $t$ lies on the unique path connecting $s$ to the root. We will also say that $t$ is an {\em{ancestor}} of $s$. Note that in this notation every node is its own descendant as well as its own ancestor.

\begin{definition}[\bf tree decomposition]
A {\em tree decomposition} of a graph $G$ is a pair $(T,\beta)$, where $T$ 
is a rooted tree and $\beta \colon V(T) \to 2^{V(G)}$ is a mapping such that:
\begin{itemize}
  \item for each node $v \in V(G)$, the set $\{t \in V(G)\ |\ v \in \beta(t)\}$ induces a nonempty and connected subtree of~$T$,
  \item for each edge $e \in E(G)$, there exists $t \in V(T)$ such that $e \subseteq \beta(t)$.
\end{itemize}
\end{definition}

Sets $\beta(t)$ for $t\in V(T)$ are called the {\em{bags}} of the decomposition, while sets $\beta(s)\cap \beta(t)$ for $st\in E(T)$ are called the {\em{adhesions}}. We sometimes implicitly identify a node of $T$ with the bag associated with it. The {\em{width}} of a tree decomposition $T$ is equal to its maximum bag size decremented by one, i.e., $\max_{t\in V(T)} |\beta(t)|-1$. The {\em{adhesion width}} of $T$ is equal to its maximum adhesion size, i.e., $\max_{st\in E(T)} |\beta(s)\cap \beta(t)|$. We define also additional mappings as follows:
\begin{align*}
\gamma(t) & = \bigcup_{u\preceq t} \beta(u), \\
\sigma(t) & = \begin{cases} \emptyset & \text{if t is the root of }T \\ \beta(t) \cap \beta(\parent(t)) & \text{otherwise,}\end{cases} \\
\alpha(t) & = \gamma(t)\setminus \sigma(t).
\end{align*}
The {\em{treewidth}} of a graph, denoted $\tw(G)$, is equal to the minimum width of its tree decomposition. Let us remark that in this paper we will be mostly working with graphs of treewidth {\em{less than}} $k$, while most of the literature on the subject considers graphs of treewidth {\em{at most}} $k$. This irrelevant detail will help us avoid clumsy additive constants in many arguments.

If $\bagfam\subseteq 2^{V(G)}$ is a family of subsets of vertices, then we say that $\bagfam$ {\em{captures}} a tree decomposition $(T,\beta)$ if $\beta(t)\in \bagfam$ for each $t\in V(T)$. In this context we often call $\bagfam$ a {\em{family of potential bags}}.

Graphs of treewidth at most $k$ are known to be {\em{$k$-degenerate}}, that is, every subgraph of a graph of treewidth at most $k$ has a vertex of degree at most $k$. This in particular implies that an $n$-vertex graph of treewidth at most $k$ can have at most $kn$ edges.

Let $G$ be a graph and let $S\subseteq V(G)$. We say that a separation $(A,B)$ in $G$ is {\em{$\alpha$-balanced for $S$}} if $|(A\setminus B)\cap S|,|(B\setminus A)\cap S|\leq \alpha|S|$. The following lemma states that graphs of bounded treewidth provide balanced separations of small order.

\begin{lemma}[\cite{Bodlaender98}]\label{lem:balanced}
Let $G$ be a graph with $\tw(G)<k$ and let $S\subseteq V(G)$ be a subset of vertices. Then there exists a $\frac{2}{3}$-balanced separation for $S$ of order at most $k$.
\end{lemma}

\paragraph*{Improved graph.} For a positive integer $k$, we say that a graph $H$ is {\em{$k$-complemented}} if the implication $(\conn_H(x,y)\geq k) \Rightarrow (xy\in E(H))$ holds for every pair of vertices $x,y\in V(H)$. For every graph $G$ we can construct a {\em{$k$-improved graph}} $\imp{G}{k}$ by having $V(\imp{G}{k})=V(G)$ and $xy\in E(\imp{G}{k})$ if and only if $\conn_G(x,y)\geq k$. Observe that $\imp{G}{k}$ is a supergraph of $G$, since $\conn_G(x,y)=+\infty$ for all $x,y$ that are adjacent in $G$. Moreover, observe that every $k$-complemented supergraph of $G$ must be also a supergraph of $\imp{G}{k}$. It appears that $\imp{G}{k}$ is $k$-complemented itself, and hence it is the unique minimal $k$-complemented supergraph of $G$.

\begin{lemma}[\cite{Bodlaender03,ClautiauxCMN03},\app]
\label{lem:improval}
For every graph $G$ and positive integer $k$, the graph $\imp{G}{k}$ is $k$-complemented. Consequently, it is the unique minimal $k$-complemented supergraph of $G$.
\end{lemma}
For completeness we give a proof of Lemma~\ref{lem:improval} in the appendix. The following lemma formally relates tree decompositions of a graph and its improved graph, and states that for the sake of computing a tree decomposition of small width we may focus on the improved graph. The proof (given in the appendix) is basically an application of a simple fact that two vertices $x,y$ with $\conn(x,y)\geq k$ must be simultaneously present in some bag of every tree decomposition of width smaller than $k$.

\begin{lemma}[\app]\label{lem:same-decomp}
Let $k$ be a positive integer and let $G$ be a graph. Then every tree decomposition $(T,\beta)$ of $G$ that has width less than $k$, is also a tree decomposition of $\imp{G}{k}$. In particular, if $\tw(G)<k$ then $\tw(G)=\tw(\imp{G}{k})$.
\end{lemma}

The idea of focusing on the improved graph dates back to the work of Bodlaender on a linear time FPT algorithm for treewidth~\cite{Bodlaender96}. Actually, Bodlaender was using a weaker variant an improved graph, where an edge is added only if the vertices in question share at least $k$ common neighbors. The main point was that this weaker variant can be computed in linear time~\cite{Bodlaender96} with respect to the size of the graph. In our work we use the stronger variant, and we can afford spending more time on computing the improved graph.

\begin{lemma}[\app]\label{lem:comp-imp}
There exists an algorithm that, given a positive integer $k$ and a graph $G$ on $n$ vertices, works in $\Oh(k^2n^3)$ time and either correctly concludes that $\tw(G)\geq k$, or computes $\imp{G}{k}$. 
\end{lemma}

\newcommand{\obj}{\mathcal{D}}
\newcommand{\strobj}{\mathcal{E}}

\paragraph{Isomorphisms and isomorphism-invariance.} We say that two graphs $G_1,G_2$ are {\em{isomorphic}} if there exists a bijection $\phi\colon V(G_1)\to V(G_2)$, called {\em{isomorphism}}, such that $xy\in E(G_1)\Leftrightarrow \phi(x)\phi(y)\in E(G_2)$ for all $x,y\in V(G_1)$. We will often say that some object $\obj(G)$ (e.g., a set of vertices, a family of sets of vertices), whose definition depends on the graph, is {\em{isomorphism-invariant}} or {\em{canonical}}. By this we mean that for any graph $G'$ that is isomorphic to $G$ with isomorphism $\phi$, it holds that $\obj(G')=\phi(\obj(G))$, where $\phi(\obj(G))$ denotes the object $\obj(G)$ with all the vertices of $G$ replaced by their images under $\phi$. The precise meaning of this term will be always clear from the context. Most often, isomorphism-invariance of some definition or of the output of some algorithm will be obvious, since the description does not depend on the internal representation of the graph, nor uses any tie-breaking rules for choosing arbitrary objects.

We also extend the notion of isomorphism-invariance to {\em{structures}}, which are formed by the considered graph $G$ together with some object $\strobj$ (e.g., a vertex, a set of vertices, a subgraph). Such structures usually represent the graph together with an initial set of choices made by some algorithm, for instance the starting vertex from which the construction of a tree decomposition begins. We say that some object $\obj(G,\strobj)$, whose definition is dependant on the structure $(G,\strobj)$, is {\em{invariant under isomorphism of $(G,\strobj)$}}, if $\obj(G',\strobj')=\phi(\obj(G,\strobj))$ for any structure $(G',\strobj')$ such that $\phi$ is an isomorphism between $G$ and $G'$ that additionally satisfies $\phi(\strobj)=\strobj'$. Again, the precise definition of this term will be always clear from the context.

\section{The case of no clique separators}\label{sec:no-cliqueseps}
Let $G$ be graph and let $S\subseteq V(G)$ be any subset of vertices. The following definition will be a crucial technical ingredient in our reasonings.

\begin{definition}
Suppose that $(S_L,S_R)$ is a partition of $S$. We say that a separation $(A,B)$ of $G$ is {\em{stable for $(S_L,S_R)$}} if $(A,B)$ is a minimum-order $S_L-S_R$ separation. A separation $(A,B)$ is {\em{$S$-stable}} if it is stable for some partition of $S$.
\end{definition}

\newcommand{\sepfam}{\mathcal{F}}

Note that $(S_L,V(G))$ and $(V(G),S_R)$ are both $S_L-S_R$ separations, and they have orders $|S_L|$ and $|S_R|$, respectively. Hence, if $(A,B)$ is a separation that is stable for $(S_L,S_R)$, then in particular $|A\cap B|\leq \min(|S_L|,|S_R|)$.

The following lemma will be the main tool for constructing an isomorphism invariant family of candidate bags.

\begin{lemma}\label{lem:magical}
Let $G$ be a graph, let $S\subseteq V(G)$ be any subset of vertices, and let $\sepfam$ be any finite family of $S$-stable separations. Define
$$X:=S\cup \bigcup_{(A,B)\in \sepfam} A\cap B.$$
Suppose that $Z$ is the vertex set of any connected component of $G\setminus X$. Then $|N(Z)|\leq |S|$.
\end{lemma}
\begin{proof}
We proceed by induction w.r.t. $|\sepfam|$. For $\sepfam=\emptyset$ we have $N(Z)\subseteq X=S$, so the claim is trivial.

Assume then that $\sepfam=\sepfam'\cup \{(A,B)\}$ for some family $\sepfam'$ with $|\sepfam'|<|\sepfam|$, and let $X'$ be defined in the same manner for $\sepfam'$ as $X$ is for $\sepfam$. As $(A,B)$ is $S$-stable, there is some a partition $(S_L,S_R)$ of $S$ such that $(A,B)$ is a minimum-order $S_L-S_R$ separation. Let $Z$ be the vertex set of any connected component of $G\setminus X$, and let $Z'\supseteq Z$ be the vertex set of the connected component of $G\setminus X'$ that contains $Z$. Since $G[Z]$ is connected and $Z\cap (A\cap B)=\emptyset$, we have that either $Z\subseteq A\setminus B$ or $Z\subseteq B\setminus A$; without loss of generality assume the former. 

Let $(C,D)=(V(G)\setminus Z',N[Z'])$. Observe that $(C,D)$ is a separation in $G$. Moreover, since $Z'\cap S=\emptyset$, then $(C,D)$ is a $S-Z'$ separation, so in particular also a $S-Z$ separation. Furthermore, observe that $C\cap D=N(Z')$, so by the induction hypothesis we obtain $|C\cap D|= |N(Z')|\leq |S|$. We can partition $V(G)$ into $9$ parts, where the part a vertex belongs to depends on its membership to $A\setminus B$, $A\cap B$, or $B\setminus A$, and its membership to $C\setminus D$, $C\cap D$, or $D\setminus C$. Let us call these parts $Q_{1,1}, Q_{1,2},\ldots,Q_{3,3}$, as depicted on Figure~\ref{fig:magical}.

\begin{figure}[htbp!]
                \centering
                \def\svgwidth{0.8\columnwidth}
                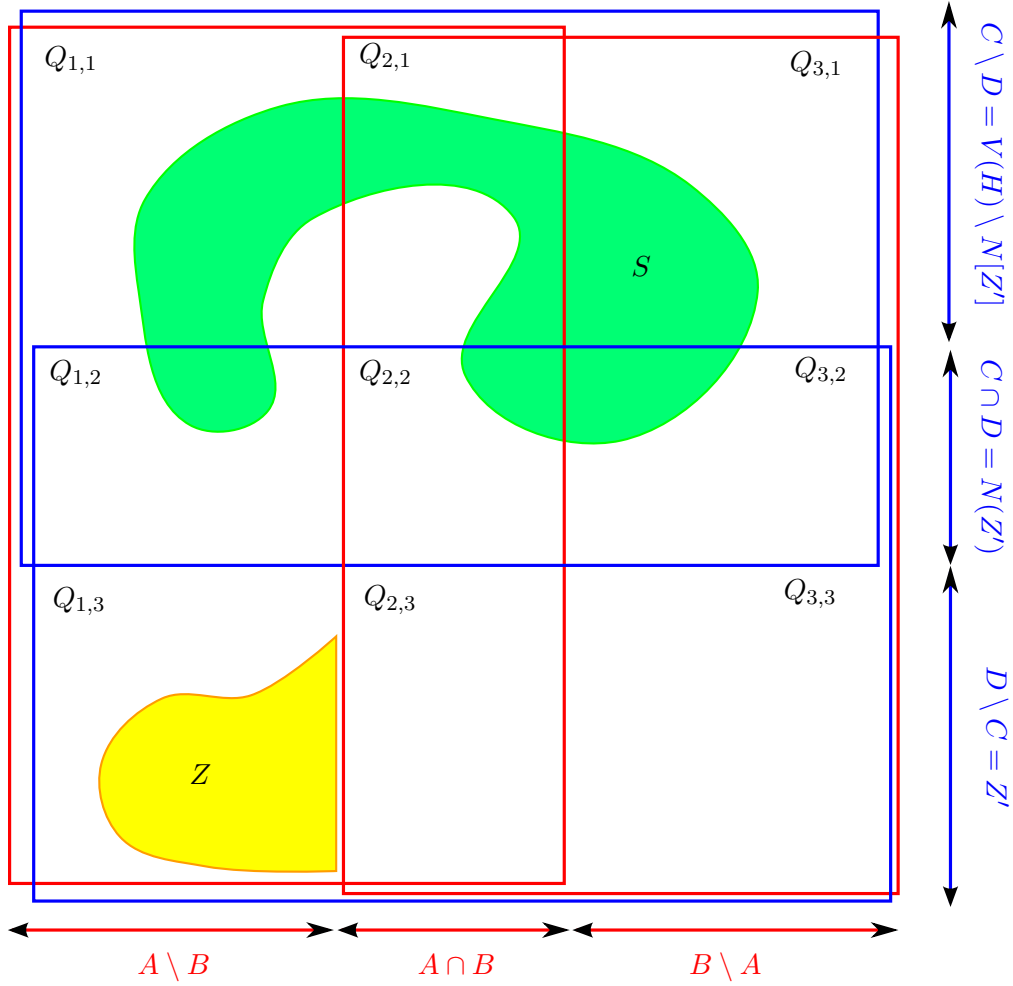
\caption{Sets in the proof of Lemma~\ref{lem:magical}}\label{fig:magical}
\end{figure}

Observe now that $Z\subseteq (A\setminus B)\cap (D\setminus C)=Q_{1,3}$. Since $X\setminus X'\subseteq A\cap B$ and $C\cap D=N(Z')\subseteq X'$, we have that $N(Z)\subseteq (A\cap C\cap D)\cup (Z'\cap A\cap B)=Q_{1,2}\cup Q_{2,2}\cup Q_{2,3}$. On the other hand, by induction hypothesis we have $|N(Z')|=|Q_{1,2}\cup Q_{2,2}\cup Q_{3,2}|\leq |S|$. Hence, to prove that $|N(Z)|\leq |S|$, it suffices to prove that $|Q_{1,2}\cup Q_{2,2}\cup Q_{2,3}|\leq |Q_{1,2}\cup Q_{2,2}\cup Q_{3,2}|$, or, equivalently, $|Q_{2,3}|\leq |Q_{3,2}|$.

To this end, consider a pair of subsets $(L,R)=(A\cup D,B\cap C)$. We first claim that $(L,R)$ is a separation. Indeed, if $u\in L\setminus R=Q_{1,1}\cup Q_{1,2}\cup Q_{1,3}\cup Q_{2,3}\cup Q_{3,3}$ and $v\in R\setminus L=Q_{3,1}$, then existence of an edge $uv$ would contradict the fact that both $(A,B)$ and $(C,D)$ are separations. Now we claim that $(L,R)$ is a $S_L-S_R$ separation. Indeed, we have $S_L\subseteq A\subseteq L$, and moreover $S_R\subseteq B$ and $S_R\subseteq S\subseteq C$ implies that $S_R\subseteq B\cap C=R$.
Since $(A,B)$ is a minimum-order $S_L-S_R$ separation, we have
$$|Q_{2,1}\cup Q_{2,2}\cup Q_{3,2}| = |L \cap R| \geq |A \cap B| = |Q_{2,1}\cup Q_{2,2}\cup Q_{2,3}|.$$
Hence indeed $|Q_{2,3}|\leq |Q_{3,2}|$ and we are done.
\end{proof}

\newcommand{\bagthresh}{\tau}
\newcommand{\bigadh}{\rho}
\newcommand{\bigbag}{\zeta}

Lemma~\ref{lem:magical} is used in the following result, which encapsulates one step of the construction of an isomorphism-invariant family of candidate bags. In the sequel, we will use the following parameters: 
\begin{eqnarray*}
\bagthresh & = & 6k, \\
\bigadh & = & \bagthresh+2(k-1)\cdot \binom{\bagthresh}{2} = \Oh(k^3), \\
\bigbag & = & \bigadh+2k\cdot\binom{\bigadh}{k+1}^2 = 2^{\Oh(k\log k)}. 
\end{eqnarray*}

\begin{lemma}\label{lem:local-decomposition}
Let $k$ be a positive integer and let $H$ be a connected graph that is $k$-complemented. Let $S\subseteq V(H)$ be a subset of vertices such that (a) $\emptyset\neq S\subsetneq V(H)$, (b) $|S|\leq \bigadh$, (c) $S$ does {\em{not}} induce a clique, (d) $H\setminus S$ is connected, and (e) $S=N_H(V(H)\setminus S)$. There exists an algorithm that either correctly concludes that $\tw(H)\geq k$, or finds a set $X$ with the following properties:
\begin{enumerate}[(i)]
\item\label{pr:progress} $X\supsetneq S$, that is, $X$ is a proper superset of $S$;
\item\label{pr:bagbound} $|X|\leq \bigbag$; and
\item\label{pr:adhbound} if $Z$ is the vertex set of any connected component of $H\setminus X$, then $|N(Z)|\leq \bigadh$.
\end{enumerate}
The algorithm runs in $2^{\Oh(k\log k)}\cdot |V(H)|$ time and the constructed set $X$ is invariant with respect to isomorphisms of the structure $(H,S)$.
\end{lemma}
\begin{proof}
We consider two cases. In the first case we assume that $|S|\leq \bagthresh$, and in the second case we assume that $\bagthresh<|S|\leq \bigadh$.

\paragraph*{Case 1: $|S|\leq \bagthresh$.} Consider any pair of vertices $x,y\in S$ such that $xy\notin E(H)$. Since $H$ is $k$-complemented, we have that $\conn_H(x,y)<k$, and hence the minimum-order separations that separate $x$ and $y$ have order less than $k$. Let $(A^x_{x,y},B^x_{x,y})$ and $(A^y_{x,y},B^y_{x,y})$ be the minimum-order separations separating $x$ and $y$ that are pushed towards $x$ and $y$, respectively. Recall that, by the definition of an $x-y$ separation, the separators $A^x_{x,y} \cap B^x_{x,y}$ and $A^y_{x,y} \cap B^y_{x,y}$ do not contain
$x$ nor $y$. We define the set $X$ as follows:
\begin{eqnarray}\label{eq:casesmall}
X := S\cup \bigcup_{\substack{x,y\in S,\\ xy\notin E(H)}} (A^x_{x,y}\cap B^x_{x,y})\cup (A^y_{x,y}\cap B^y_{x,y}).
\end{eqnarray}
In other words, we enhance $S$ by adding, for every pair of nonadjacent vertices from $S$, both the extreme minimum separators separating them. Observe that the definition of $X$ is invariant with respect to isomorphisms of the structure $(H,S)$. Observe also that $|X|\leq |S|+2(k-1)\cdot \binom{|S|}{2}\leq \bigadh$, since $|S|\leq \bagthresh$. This proves properties (\ref{pr:bagbound}) and (\ref{pr:adhbound}) of $X$. For property (\ref{pr:progress}), observe that by our assumptions about the set $S$ we have that there exists at least one pair $x,y\in S$ with $xy\notin E(H)$ (by properties (a) and (c)), and that for this pair there exists a path that starts in $x$, ends in $y$, and whose all internal vertices are contained in $V(H)\setminus S$ (by properties (a), (d), and (e)). The internal vertices of this path need to include a vertex of $A^x_{x,y}\cap B^x_{x,y}$ and a vertex $A^y_{x,y}\cap B^y_{x,y}$. Hence, the set $X$ contains at least one vertex outside $S$, and property (\ref{pr:progress}) holds.

Note that in this case $X$ can be computed in time $k^{\Oh(1)}\cdot |V(H)|$, by considering every pair of non-adjacent vertices of $S$,
and for each of them running at most $k$ iterations of the Ford-Fulkerson algorithm.
(Observe that, unless $\tw(H) \geq k$, we have $|E(H)| = \Oh(k|V(H)|)$ and, hence, each of these iterations takes $\Oh(k|V(H)|)$ time.)

\paragraph*{Case 2: $\bagthresh<|S|\leq \bigadh$.} We construct the set $X$ as follows. Consider all the pairs $(L,R)$ of subsets of $S$ such that $L\cap R=\emptyset$ and $|L|=|R|=k+1$. For every such pair, let us verify whether the minimum order of a separation separating $L$ and $R$ is at most $k$, and in this case let us compute minimum-order $L-R$ separations $(A^L_{L,R},B^L_{L,R})$, $(A^R_{L,R},B^R_{L,R})$ that are pushed towards $L$ and $R$, respectively (note that here the vertices of $L$ and $R$ can be included in the separator). We now define $X$ similarly as in the previous case:
\begin{eqnarray}\label{eq:casebig}
X := S\cup \bigcup_{\substack{L,R\subseteq S,\ L\cap R=\emptyset,\\ |L|=|R|=k+1,\ \conn(L,R)\leq k}} (A^L_{L,R}\cap B^L_{L,R})\cup (A^R_{L,R}\cap B^R_{L,R}).
\end{eqnarray}
Note that the definition of $X$ is invariant with respect to isomorphism of the structure $(H,S)$. Property (\ref{pr:bagbound}) of $X$ follows directly from the definition and the fact that $|S|\leq \bigadh$. We proceed to checking the other two properties.

For property (\ref{pr:adhbound}), take any pair $(L,R)$ considered in the union in (\ref{eq:casebig}), and let us look at a separation $(A,B)\in \{(A^L_{L,R},B^L_{L,R}),(A^R_{L,R},B^R_{L,R})\}$. We can easily construct a partition $(S_L,S_R)$ of $S$ such that $L\subseteq S_L\subseteq A$ and $R\subseteq S_R\subseteq B$, by assigning vertices of $L$ to $S_L$, vertices of $R$ to $S_R$, and assigning vertices of $S\setminus (L\cup R)$ according to their containment to $A$ or $B$ (thus, we have a unique choice for all the vertices apart from $(S\setminus (L\cup R))\cap (A\cap B)$). Then $(A,B)$ is a $S_L-S_R$ separation, and since it was a minimum-order $L-R$ separation, it must be also a minimum-order $S_L-S_R$ separation. Hence, $(A,B)$ is an $S$-stable separation. We infer that all the separations considered in the union in (\ref{eq:casebig}) are $S$-stable. From Lemma~\ref{lem:magical} it follows that $|N(Z)|\leq |S|\leq \bigadh$, where $Z$ is the vertex set of any connected component of $H\setminus X$. This concludes the proof of property (\ref{pr:adhbound}).

For property (\ref{pr:progress}), if $\tw(H) < k$ then Lemma~\ref{lem:balanced} implies an existence of a separation $(A,B)$ of order at most $k$ that is $\frac{2}{3}$-balanced for $S$.
Consequently,
$$|(A\setminus B)\cap S|\geq |S|-|A\cap B|-|(B\setminus A)\cap S|\geq \frac{1}{3}|S|-k>k,$$
where the last inequality follows from the fact that $|S|>\bagthresh=6k$. Symmetrically, $|(B\setminus A)\cap S|>k$. Let then $L$ and $R$ be any two subsets of $(A\setminus B)\cap S$ and $(B\setminus A)\cap S$, respectively, that have sizes $k+1$. Observe that the existence of separation $(A,B)$ certifies that $\conn(L,R)\leq k$, and hence the pair $(L,R)$ is considered in the union in (\ref{eq:casebig}). Recall that $(A^L_{L,R},B^L_{L,R})$ is then the corresponding $L-R$ separation of minimum order that is pushed towards $L$. Since $|A^L_{L,R}\cap B^L_{L,R}|\leq k$ and $|L|,|R|=k+1$, there exist some vertices $x,y$ such that $x\in L\setminus (A^L_{L,R}\cap B^L_{L,R})$ and $y\in R\setminus (A^L_{L,R}\cap B^L_{L,R})$. Similarly as in Case 1, there exists a path that starts in $x$, ends in $y$, and whose all internal vertices are contained in $V(H)\setminus S$ (by properties (a), (d), and (e)). The internal vertices of this path need to contain at least one vertex from $A^L_{L,R}\cap B^L_{L,R}$. We infer that if $\tw(H)<k$, then $X$ computed using formula~(\ref{eq:casebig}) is a proper superset of $S$. Therefore, we may output the conclusion that $\tw(H)\geq k$ if $X$ computed using formula~(\ref{eq:casebig}) is equal to $S$; otherwise $X$ is a proper superset of $S$, as requested.

Note that in this case $X$ can be computed in time $2^{\Oh(k\log k)}\cdot |V(H)|$ as follows. There are at most $\binom{\bigadh}{k+1}^2=2^{\Oh(k\log k)}$ possible choices for $L$ and $R$, and for each of them we may check whether $\conn(L,R)\leq k$ (and compute $(A^L_{L,R},B^L_{L,R})$ and $(A^R_{L,R},B^R_{L,R})$, if needed) using at most $k+1$ iterations of the Ford-Fulkerson algorithm. Similarly as in Case 1, we may assume that each iteration takes $\Oh(k|V(H)|)$ time.
\end{proof}

Armed with Lemma~\ref{lem:local-decomposition}, we may proceed to the main result of this section, that is, the enumeration of an isomorphism-invariant family of bags that captures a tree decomposition of reasonably small width.

\begin{theorem}\label{thm:no-clique-seps}
Let $k$ be a positive integer, and let $G$ be a graph on $n$ vertices that is clique-separator free (in particular, connected), and $k$-complemented. There exists an algorithm that computes an isomorphism-invariant family of potential bags $\bagfam\subseteq 2^{V(G)}$ with the following properties:
\begin{enumerate}[(i)]
\item\label{pr:small-bags} $|B|\leq \bigbag$ for each $B\in \bagfam$;
\item\label{pr:small-fam} $|\bagfam| = \Oh(n^2)$;
\item\label{pr:capture} assuming that $\tw(G)<k$, the family $\bagfam$ captures some tree decomposition of $G$ that has width at most $\bigbag+1 = 2^{\Oh(k\log k)}$ and adhesion width at most $\bigadh = \Oh(k^3)$.
\end{enumerate}
Moreover, the algorithm runs in $2^{\Oh(k\log k)}\cdot n^3$ time.
\end{theorem}
\begin{proof}
We first consider the border case when $G$ is a clique. Then we can output $\bagfam=\{V(G)\}$ if $n\leq k$, and $\bagfam=\emptyset$ if $n>k$. In the following we assume that $G$ is not a clique.

Let $u$ be any vertex of $G$ whose degree is less than $k$. Observe that since graphs of treewidth $t$ are $t$-degenerate, then there exists at least one such vertex, provided that $\tw(G)<k$ (otherwise we may output $\bagfam=\emptyset$). For each such vertex $u$ we will construct a family $\bagfam_u$ such that (a) $\bagfam_u$ satisfies properties (\ref{pr:small-bags}) and (\ref{pr:capture}), (b) $|\bagfam_u| = \Oh(n)$, (c) the definition of $\bagfam_u$ is invariant with respect to isomorphisms of the structure $(G,u)$, and (d) computing $\bagfam_u$ can be done in time $2^{\Oh(k\log k)}\cdot n^2$. The final family $\bagfam$ can be then defined simply as:
\begin{eqnarray*}
\bagfam := \bigcup_{\substack{u\in V(G)\\ \deg(u)<k}} \bagfam_u.
\end{eqnarray*}
It follows that the definition of $\bagfam$ is isomorphism-invariant, and the running time of the whole algorithm follows from the running time of computing a single $\bagfam_u$.

Let us focus on one vertex $u$. The algorithm will actually compute some tree decomposition $(T_u,\beta_u)$ of $G$ that has width $2^{\Oh(k\log k)}$ and adhesion width $\Oh(k^3)$, and then it will output the set of all its bags as $\bagfam_u$. Hence it will be clear that $(T_u,\beta_u)$ is captured by $\bagfam_u$. The definition of $(T_u,\beta_u)$ will be invariant with respect to isomorphisms of the structure $(G,u)$.

We describe the algorithm as a recursive routine that takes as input an induced subgraph $G'$ of $G$ together with a set $X$, $\emptyset\neq X\subseteq V(G')$, that is supposed to be the root bag. We ensure that the algorithm does not loop using a potential $\pot(G',X)=(|V(G')|+1)\cdot (\bigbag+1)-|X|$. Formally, every call of the algorithm will use only calls for inputs with a strictly smaller potential, and the potential is always nonnegative. For the set $X$, we require the following invariants:
\begin{enumerate}[(a)]
\item\label{inv:smallbag} $|X|\leq \bigbag$, 
\item\label{inv:sep} $X\supseteq N_G(V(G')\setminus X)$, that is, $X$ separates $V(G')\setminus X$ from the rest of the graph, 
\item\label{inv:smalladh} $|N_G(Z)|\leq \bigadh$ for $Z$ being the vertex set of any connected component of $G'\setminus X$, and 
\item\label{inv:u} $u\notin N_G[V(G')\setminus X]$.
\end{enumerate}
The output of the algorithm will be a tree decomposition of $G'$ that has $\Oh(|V(G')|)$ bags, width $2^{\Oh(k\log k)}$ and adhesion width $\Oh(k^3)$, and whose root bag is equal to $X$. The top-most call of the routine is for $G'=G$ and $X=N[u]$; that is, we will construct a tree decomposition of $G$ with the top bag being $N[u]$. Note that the invariants (\ref{inv:smallbag}), (\ref{inv:sep}), (\ref{inv:smalladh}), and (\ref{inv:u}) are satisfied in this call.

The algorithm first checks whether $X=V(G')$, in which case it simply returns a decomposition consisting of one root bag being $X$. Assume then that $X\subsetneq V(G')$. The algorithm considers all the connected components of $G'\setminus X$. Let $Z$ be the vertex set of one of them; recall that $|N(Z)|\leq \bigadh$. Let us examine the graph $G''_Z=G'[N[Z]]$, and let $S_Z=N(Z)$. Observe that the pair $(G''_Z,S_Z)$ satisfies the prerequisites of Lemma~\ref{lem:local-decomposition}:
\begin{itemize}
\item prerequisite (a) follows from the fact that $Z$ is nonempty and $G$ is connected;
\item prerequisite (b) follows from invariant (\ref{inv:smalladh});
\item prerequisite (c) follows from the observation that if $S_Z$ induced a clique, then this clique would be a clique separator separating any vertex of $Z$ from $u$ (since invariants (\ref{inv:sep}) and (\ref{inv:u}) are satisfied);
\item prerequisite (d) follows from the fact that $Z$ is connected;
\item prerequisite (e) follows from the definition of $S_Z$.
\end{itemize}
Hence, let us apply the algorithm of Lemma~\ref{lem:local-decomposition} to the pair $(G''_Z,S_Z)$, obtaining a set $X_Z$. 
(If the algorithm of Lemma~\ref{lem:local-decomposition} returned that $\tw(G''_Z)\geq k$, then also $\tw(G)\geq k$ and we may return $\bagfam=\emptyset$.)

Having computed $X_Z$, the algorithm calls itself recursively for the graph $G''_Z$ and the top bag $X_Z$. Note that by the properties guaranteed by Lemma~\ref{lem:local-decomposition}, we either have that $|V(G''_Z)|<|V(G')|$ or that $G'=G''_Z$ and $|X_Z|>|X|$. Hence we have that $\pot(G''_Z,X_Z)<\pot(G',X)$, as was requested. Let us check also that the invariants are satisfied for this call: invariant (\ref{inv:sep}) follows from the definitions of $G''_Z$ and $S_Z$, invariants (\ref{inv:smallbag}) and (\ref{inv:smalladh}) follow directly from Lemma~\ref{lem:local-decomposition}, and invariant (\ref{inv:u}) follows from the definition of $(G''_Z,X_Z)$, and satisfaction of this invariant for the call $(G',X)$.

Let $(T_Z,\beta_Z)$ be the returned tree decomposition of $G''_Z$; this decomposition has guaranteed small width and adhesion width, and its top bag is $X_Z$. We construct the final output decomposition of $G'$ by creating a root bag equal to $X$, and attaching all the decompositions $(T_Z,\beta_Z)$ for connected components $Z$ as subtrees below this bag. It is easy to verify that this is indeed a tree decomposition of the graph $G'$. Moreover, it has required width and adhesion width; note here that adhesions adjacent to the root bag are neighbourhoods of connected components of $G'\setminus X$, which have size at most $\bigadh$ by invariant (\ref{inv:smalladh}).

Observe also that since we start with the top-most call $(G,N[u])$, and operations performed in each recursive call $(G',X)$ are invariant with respect to isomorphisms of the structure $(G',X)$, it follows that the computed tree decomposition of $G$ is invariant with respect to isomorphisms of the structure $(G',u)$.

We are left with bounding the number of bags of the returned decomposition and analyzing the running time of the algorithm. Let $(T,\beta)$ be the returned tree decomposition of $G$. Recall that for every $v\in V(G)$, the set of nodes whose bags contain $v$ induces a connected subtree of $T$. Let $t(v)$ be the top-most node of this subtree, i.e., the top-most node where $v$ appears in the bag. We will also say that vertex $v$ {\em{charges}} node $t(v)$. We now claim that every node of the decomposition $(T,\beta)$ is charged at least once. Indeed, the root node with bag $N[u]$ is charged by $u$, and for every other node, its bag appears as $X_Z$ in some recursive call $(G',X)$. By Lemma~\ref{lem:local-decomposition}, property~(\ref{pr:progress}), we have $X_Z\setminus X\neq \emptyset$, and so the considered node will be charged by any vertex of $X_Z\setminus X$. Consequently, the total number of produced nodes is at most the total number of vertices of the graph, which is $n$.

To bound the running time of the algorithm, we sum the total work performed in each recursive call of the algorithm. Let us take one call, say $(G',X)$, and observe that the work performed in this call (excluding recursive sub-calls) consists of applications of the algorithm of Lemma~\ref{lem:local-decomposition} to instances $(G''_Z,S_Z)$, for all connected components $Z$ of $G'\setminus X$. By Lemma~\ref{lem:local-decomposition}, each such application (together with auxiliary operations like construction of the subinstance) takes at most $2^{\Oh(k\log k)}\cdot n$ time, and results in obtaining one subtree of the decomposition that is then attached below the bag $X$. Let us assign the time spent on running the algorithm of Lemma~\ref{lem:local-decomposition} on instance $(G''_Z,S_Z)$ to the root node of this subtree, i.e., to the node with bag $X_Z$. Observe that thus every node of the constructed tree decomposition $(T,\beta)$ of $G$ is being assigned at most once. Since the total number of nodes in $T$ is at most $n$, we infer that the total time used by the algorithm is $2^{\Oh(k\log k)}\cdot n^2$, as requested.
\end{proof}

\section{Taming clique separators}\label{sec:cliqueseps}
The main tool that we will use in this section is a decomposition theorem that breaks the graph using minimal clique separators into pieces that cannot be decomposed further. It appears that such a decomposition can be done, with a set of its bags defined in a unique manner. The idea of decomposing a graph using clique separators dates back to the work of Tarjan~\cite{Tarjan85}, and has been studied intensively thereafter. We refer to an introductory article of Berry et al.~\cite{BerryPS10} for more details. The following theorem states all the properties of this decomposition in the language of tree decompositions.

\begin{theorem}[see e.g.~\cite{BerryPS10}]\label{thm:cliqsepdecomp}
Let $G$ be a connected graph. There exists a tree decomposition $(T^\star,\beta^\star)$ of $G$, called {\em{clique minimal separator decomposition}}, with the following properties:
\begin{itemize}
\item for every $t\in V(T^\star)$, $G[\beta^\star(t)]$ is clique-separator free;
\item each adhesion of $(T^\star,\beta^\star)$ is a clique in $G$.
\end{itemize}
Moreover, $T^\star$ has at most $n-1$ nodes, and the bags of $(T^\star,\beta^\star)$ are exactly all the inclusion-wise maximal induced subgraphs of $G$ that are clique-separator free. Consequently, the family of bags of $(T^\star,\beta^\star)$ is isomorphism-invariant. Finally, the decomposition $(T^\star,\beta^\star)$ can be computed in $\Oh(nm)$ time.
\end{theorem}

We remark that the exact shape of the clique minimal separator decomposition is not isomorphism-invariant: the construction procedure depends on the order in which inclusion-wise minimal clique separators of the graph are considered. However, the family of bags of this decomposition is isomorphism-invariant, since these bags may be characterized as all the inclusion-wise maximal induced subgraphs of $G$ that are clique-separator free. In other words, all the possible runs of the decomposition algorithm yield the same family of bags, just arranged in a different manner.

Theorem~\ref{thm:cliqsepdecomp} enables us to conveniently extend Theorem~\ref{thm:no-clique-seps} to graphs that may contain clique separations.

\begin{theorem}\label{thm:clique-seps}
Let $k$ be a positive integer, and let $G$ be a graph on $n$ vertices that is $k$-complemented. There exists an algorithm that computes an isomorphism-invariant family of bags $\bagfam$ with the following properties:
\begin{enumerate}[(i)]
\item\label{pr1cs} $|B|\leq \bigbag$ for each $B\in \bagfam$;
\item\label{pr2cs} $|\bagfam|\leq \Oh(k^2n^2)$;
\item\label{pr3cs} assuming that $\tw(G)<k$, the family $\bagfam$ captures some tree decomposition of $G$ that has width at most $\bigbag+1 = 2^{\Oh(k\log k)}$ and adhesion width at most $\bigadh = \Oh(k^3)$.
\end{enumerate}
Moreover, the algorithm runs in $2^{\Oh(k\log k)}\cdot n^3$ time.
\end{theorem}
\begin{proof}
In $\Oh(nm)$ time we compute a clique minimal separator decomposition $(T^\star,\beta^\star)$ of $G$. We can assume that all the adhesions of $(T^\star,\beta^\star)$ are of size at most $k$, since otherwise $G$ contains a clique on $k+1$ vertices; then, $\tw(G)\geq k$ and we may output $\bagfam=\emptyset$. In the following, let $r^\star$ be the root of $T^\star$. Observe that:
$$\sum_{t\in V(T^\star)} |\beta^\star(t)|=n+\sum_{t\in V(T^\star)\setminus \{r^\star\}} |\sigma^\star(t)|\leq n+(n-2)k = \Oh(kn).$$
For every $t\in V(T^\star)$, let us examine the graph $G[\beta^\star(t)]$. Since $G[\beta^\star(t)]$ does not admit a clique separation, it is in particular connected. Moreover, since $G$ is $k$-complemented, then so is $G[\beta^\star(t)]$. Hence, $G[\beta^\star(t)]$ satisfies the prerequisites of Theorem~\ref{thm:no-clique-seps}.

For each $t\in V(T^\star)$, let us apply the algorithm of Theorem~\ref{thm:no-clique-seps} with parameter $k$ to the graph $G[\beta^\star(t)]$. Let $\bagfam_t$ be the obtained family of bags. We now define the output family to be simply $\bagfam:=\bigcup_{t\in V(T^\star)} \bagfam_t$. Observe that since the family of bags of $(T^\star,\beta^\star)$ is isomorphism-invariant, and for each $t\in V(T^\star)$ the constructed family $\bagfam_t$ is invariant with respect to isomorphisms of the graph $G[\beta^\star(t)]$, then it follows that the definition of $\bagfam$ is isomorphism-invariant. 

We now verify the requested properties of family $\bagfam$. Property (\ref{pr1cs}) follows directly from the construction and Theorem~\ref{thm:no-clique-seps}. For property (\ref{pr2cs}), by Theorem~\ref{thm:no-clique-seps} we have that $|\bagfam_t|\leq \Oh(|\beta^\star(t)|^2)$ for each $t\in V(T)$. Since $\sum_{t\in V(T^\star)} |\beta^\star(t)|=\Oh(kn)$, it follows that $|\bagfam|\leq \Oh(k^2n^2)$. Note here that a similar argument gives the bound on the running time of the algorithm: processing graph $G[\beta^\star(t)]$ takes $2^{\Oh(k\log k)}\cdot |\beta^\star(t)|^3$ time, so the whole algorithm may be implemented in time $\Oh(nm)+2^{\Oh(k\log k)}\cdot k^3n^3 = 2^{\Oh(k\log k)}\cdot n^3$.

We are left with verifying property (\ref{pr3cs}). Assume that $\tw(G)<k$; then also $\tw(G[\beta^\star(t)])<k$ for each $t\in V(T^\star)$. By Theorem~\ref{thm:no-clique-seps}, for each $t\in V(T^\star)$ there exists some tree decomposition $(T_t,\beta_t)$ of $G[\beta^\star(t)]$ that is captured by $\bagfam_t$ and satisfies properties (\ref{pr:small-bags}), (\ref{pr:small-fam}), and (\ref{pr:capture}) of Theorem~\ref{thm:no-clique-seps}. Since $\sigma^\star(t)$ is a clique for each $t\in V(T^\star)$, it follows that some bag of $(T_t,\beta_t)$ contains the whole $\sigma^\star(t)$. By re-rooting the decomposition $(T_t,\beta_t)$ if necessary, without loss of generality we may assume that $\beta_t(r_t)\supseteq \sigma^\star(t)$ for each $t\in V(T^\star)$, where $r_t$ is the root of $T_t$.

Now the goal is to combine all the decompositions $(T_t,\beta_t)$ into one decomposition $(T,\beta)$ of $G$ that satisfies the conditions expressed in property (\ref{pr3cs}). We construct $(T,\beta)$ by a bottom-up induction on decomposition $(T^\star,\beta^\star)$. For each $t\in V(T^\star)$ we will construct a decomposition $(T'_t,\beta'_t)$ of $G[\gamma^\star(t)]$ with the property that $\sigma^\star(t)$ will be contained in the root bag of $(T'_t,\beta'_t)$. Then we will simply take $(T,\beta):=(T'_{r^\star},\beta'_{r^\star})$. 

Assume we are considering a node $t\in V(T^\star)$, and assume that we have constructed decompositions $\{(T'_{t_i},\beta'_{t_i})\}_{1\leq i\leq p}$ for the children $t_1,t_2,\ldots,t_p$ of $t$ (possibly $p=0$). For each $i=1,2,\ldots,p$, recall that $\sigma^\star(t_i)$ is a clique, and hence there exists some node $s_i$ of $(T_t,\beta_t)$ whose bag contains the whole set $\sigma^\star(t_i)$. We construct 
$(T'_t,\beta'_t)$ from $(T_t,\beta_t)$ by attaching, for every $i=1,\ldots,p$, the decomposition $(T'_{t_i},\beta'_{t_i})$ as a subtree below node $s_i$. Since $\sigma^\star(t_i)$ is exactly the intersection of $\beta_t(s_i)$ and the root bag of $(T'_{t_i},\beta'_{t_i})$, it can be easily verified that $(T'_t,\beta'_t)$ constructed in this manner is a tree decomposition of $G[\gamma^\star(t)]$. Moreover, since $\sigma^\star(t)\subseteq \beta_t(r_t)$, then the invariant that $\sigma^\star(t)$ is contained in the root bag of $(T'_t,\beta'_t)$ is preserved.

The bound on the width of $(T,\beta)$ follows from the bound on the widths of decompositions $(T_t,\beta_t)$ given by Theorem~\ref{thm:no-clique-seps}. For the adhesion width, observe that the only adhesions that were not present in some decomposition $(T_t,\beta_t)$ are the adhesions created when attaching some $(T'_{t_i},\beta'_{t_i})$ below the node $s_i$. However, these adhesions are exactly adhesions of decomposition $(T^\star,\beta^\star)$, which are of size at most $k<\bigadh$. Finally, observe that each bag of $(T,\beta)$ originates in decomposition $(T_t,\beta_t)$ for some $t\in V(T^\star)$; since $(T_t,\beta_t)$ was captured by $\bagfam_t$, it follows that $(T,\beta)$ is captured by $\bagfam$.
\end{proof}

\section{Reducing bag sizes}\label{sec:redsizes}
\begin{definition}
Let $G$ be a graph, let $\bagfam\subseteq 2^{V(G)}$ be a family of candidate bags, and let $q$ be a positive integer. Then 
\begin{align*}
\bagfam^{\leq q} := \{ X\subseteq V(G)\ :\ |X|\leq q \text{ and } \exists_{B\in \bagfam} X\subseteq B\}.
\end{align*}
\end{definition}

Note that if family $\bagfam$ is isomorphism-invariant, then so does $\bagfam^{\leq q}$.

The following lemma will be the crucial technical insight of this section. Intuitively it states that by focusing on the family $\bagfam^{\leq q}$ for large enough $q$, instead of the original $\bagfam$, we still capture some tree decomposition of the graph that has a reasonably small width. The crucial point here is that the candidate bags of $\bagfam^{\leq q}$ are much smaller than those of $\bagfam$. Since the canonization algorithm of Section~\ref{sec:canon} is essentially considering all permutations of all the bags, reducing the bag size will be useful for speeding it up. 

\begin{lemma}\label{lem:Bq-captures}
Let $G$ be a connected graph of treewidth less than $k$, and let $\bagfam\subseteq 2^{V(G)}$ be a family of candidate bags that captures some tree decomposition of $G$ that has width at most $k'$ and adhesion width at most $\ell$, where $k\leq \ell\leq k'$. Then the family $\bagfam^{\leq (k+1)\ell}$ captures some tree decomposition of $G$ that has width at most $(k+1)\ell-1$.
\end{lemma}

Before we proceed with the proof of Lemma~\ref{lem:Bq-captures}, we need one more definition.

\begin{definition}[\bf connectivity-sensitive tree decomposition]
We say that a tree decomposition $(T,\beta)$ of a connected graph $G$ is {\em{connectivity-sensitive}} ({\em{cs-tree decomposition}}, for short), if the following conditions are satisfied for every $t\in V(T)$:
\begin{itemize}
\item $G[\alpha(t)]$ is connected, and
\item $\sigma(t)=N_G(\alpha(t))$.
\end{itemize}
\end{definition}

Actually, one can see that the tree decomposition constructed in the proof of Theorem~\ref{thm:no-clique-seps} is connectivity sensitive, so the family $\bagfam$ actually captures a cs-tree decomposition of the graph with required width and adhesion width. A similar conclusion, however, is not so clear in the case of Theorem~\ref{thm:clique-seps}. Fortunately, it is easy to see that every tree decomposition of a connected graph can be turned into a cs-tree decomposition without increasing the widths. For the sake of completeness, we attach the easy proof in the appendix.

\begin{lemma}[\app]\label{lem:sensitivitation}
If a connected graph $G$ admits a tree decomposition $(T,\beta)$ of width $k$ and adhesion width $\ell$, then $G$ admits also a cs-tree decomposition $(T',\beta')$ of width at most $k$ and adhesion width at most $\ell$. Moreover, every bag appearing in $(T',\beta')$ is a subset of some bag of $(T,\beta)$.
\end{lemma}

We are ready to proceed to the proof of Lemma~\ref{lem:Bq-captures}.

\begin{proof}[Proof of Lemma~\ref{lem:Bq-captures}]
Let $(T_0,\beta_0)$ be a tree decomposition of $G$ that has width at most $k'$ and adhesion width at most $\ell$, and is captured by $\bagfam$. Let $(T,\beta)$ be the cs-tree decomposition of $G$ given by Lemma~\ref{lem:sensitivitation} applied to $(T_0,\beta_0)$. That is, $(T,\beta)$ is connectivity-sensitive, has width at most $k'$ and adhesion width at most $\ell$, and its  every  bag is contained in some bag $(T_0,\beta_0)$, so in particular in some bag of $\bagfam$. We now prove that there exists a tree decomposition $(T',\beta')$ of $G$ such that:
\begin{itemize}
\item $(T',\beta')$ has width at most $(k+1)\ell-1$,
\item every bag of $(T',\beta')$ is a subset of some bag of $(T,\beta)$.
\end{itemize}
From these properties it follows that $(T',\beta')$ is captured by $\bagfam^{\leq k\ell}$, which will conclude the proof.

We construct the decomposition $(T',\beta')$ by a bottom-up induction on the decomposition $(T,\beta)$. For each node $t\in V(T)$, we construct a decomposition $(T'_t,\beta'_t)$ that has the following properties:
\begin{enumerate}[(i)]
\item\label{pr:smw} $(T'_t,\beta'_t)$ is a tree decomposition of $G[\gamma(t)]$ of width at most $(k+1)\ell-1$;
\item\label{pr:subset} every bag of $(T'_t,\beta'_t)$ is a subset of some bag of $(T,\beta)$;
\item\label{pr:root} the root bag of $(T'_t,\beta'_t)$ contains $\sigma(t)$.
\end{enumerate}
The decomposition $(T',\beta')$ will be then simply $(T'_r,\beta'_r)$, where $r$ is the root node of $(T,\beta)$.

Take any node $t$, and let $t_1,t_2,\ldots,t_p$ be its children in $T$ (possibly $p=0$ if $t$ is a leaf). By induction hypothesis we have decompositions $\{(T'_{t_i},\beta'_{t_i})\}_{1\leq i\leq p}$  for the subtrees below $t$ that satisfy properties (\ref{pr:smw}), (\ref{pr:subset}), (\ref{pr:root}).

Let us construct a graph $H_t$ from $G[\gamma(t)]$ by contracting, for every $i\in \{1,2,\ldots,p\}$, the subgraph $G[\alpha(t_i)]$ into a single vertex $u_i$; recall that this subgraph is connected since $(T,\beta)$ is connectivity-sensitive. Moreover, by connectivity-sensitivity we have that $N_{H_t}(u_i)=\sigma(t_i)$. Since $H_t$ is a minor of $G$, we infer that $\tw(H_t)\leq \tw(G)<k$. Let then $(T_{H_t},\beta_{H_t})$ be a tree decomposition of $H_t$ of width less than $k$.

We construct decomposition $(T'_t,\beta'_t)$ from $(T_{H_t},\beta_{H_t})$ in the following steps:
\begin{enumerate}
\item\label{st1} Include all the vertices of $\sigma(t)$ into each bag of $(T_{H_t},\beta_{H_t})$.
\item\label{st2} Replace every occurrence of each vertex $u_i$ in each bag by all the vertices of $N_{H_t}(u_i)=\sigma(t_i)$.
\item\label{st3} For every $i=1,2,\ldots,p$, find any node of the decomposition whose bag originally contained $u_i$ (and so now it contains $N_{H_t}(u_i)=\sigma(t_i)$). Attach the decomposition $(T'_{t_i},\beta'_{t_i})$ as a subtree below this node.
\end{enumerate}
It is straightforward to verify that $(T'_t,\beta'_t)$ constructed in this manner is a valid tree decomposition of $G[\gamma(t)]$. Moreover, observe that the bags of $(T'_t,\beta'_t)$ are of size at most $k\ell$: For bags originating in decompositions $(T'_{t_i},\beta'_{t_i})$ this follows from induction hypothesis, while bags originating in $(T_{H_t},\beta_{H_t})$ had size at most $k$ in the beginning, then got augmented by at most $\ell$ vertices in step (\ref{st1}), and finally some of the original vertices got replaced by $\ell$ other vertices in step (\ref{st2}). This implies that these bags have size at most $k\ell+\ell=(k+1)\ell$ at the end. This proves property (\ref{pr:smw}). Properties (\ref{pr:subset}) and (\ref{pr:root}) follow directly from the construction: every bag originating in $(T_{H_t},\beta_{H_t})$ is a subset of $\beta(t)$ and a superset of $\sigma(t)$, and property (\ref{pr:subset}) for bags originating in decompositions $(T'_{t_i},\beta'_{t_i})$ follows from the induction hypothesis.
\end{proof}

Using Lemma~\ref{lem:Bq-captures}, we can further refine Theorem~\ref{thm:clique-seps}. The new property (\ref{pr4}) is a technical condition that will be used later.

\begin{theorem}\label{thm:redsizes}
Let $k$ be a positive integer, and let $G$ be a graph on $n$ vertices that is connected and $k$-complemented. There exists an algorithm that computes an isomorphism-invariant family of bags $\bagfam$ with the following properties:
\begin{enumerate}[(i)]
\item\label{pr1} $|B|\leq (k+1)\bigadh\in \Oh(k^4)$ for each $B\in \bagfam$;
\item\label{pr2} $|\bagfam|\leq 2^{\Oh(k^5\log k)}\cdot n^2$;
\item\label{pr3} assuming that $\tw(G)<k$, the family $\bagfam$ captures some tree decomposition of $G$ that has width at most $(k+1)\bigadh-1\in \Oh(k^4)$;
\item\label{pr4} family $\bagfam$ is closed under taking subsets.
\end{enumerate}
Moreover, the algorithm runs in $2^{\Oh(k^5\log k)}\cdot n^3$ time.
\end{theorem}
\begin{proof}
We run the algorithm of Theorem~\ref{thm:clique-seps} on the graph $G$ to obtain an isomorphism-invariant family $\bagfam_0$. Then, we output the family $\bagfam:=\bagfam_0^{\leq (k+1)\bigadh}$. Observe that since $|B|\leq \bigbag$ for each $B\in \bagfam_0$, then each $B\in \bagfam_0$ gives rise to at most $\sum_{i=0}^{(k+1)\bigadh} \binom{\bigbag}{i}\in 2^{\Oh(k^5\log k)}$ sets in the output family $\bagfam$. This justifies the bound on $|\bagfam|$ (property (\ref{pr2})) and on the running time. Properties (\ref{pr1}) and (\ref{pr4}) follow directly from the construction, and property (\ref{pr3}) follows from Lemma~\ref{lem:Bq-captures}.
\end{proof}

\section{Canonization}\label{sec:canon}
In this section we utilize the isomorphism-invariant family of candidate bags constructed in Theorems~\ref{thm:no-clique-seps},~\ref{thm:clique-seps}, and~\ref{thm:redsizes} to give a canonization algorithm for graphs of bounded treewidth running in FPT time. 
Our main goal is to prove Theorem~\ref{main:canonization}; we then deduce Theorems~\ref{main:isomorphism} and~\ref{simple:canonization}
in Section~\ref{sec:canon-cors}.

First we introduce a concept that we call {\em{construction terms}}, which is an alternative definition of treewidth and tree decompositions via graph grammars. Our canonization algorithm then produces a canonical expression (construction term) that builds the graph; the isomorphism tests boils down to verifying equality of these canonical expressions.

\subsection{Construction terms}\label{sec:terms}
 The formalization given in this section has been known in the graph grammar literature from eighty's.  We refer to~\cite{Bodlaender98,BevernFGR13} for a review on 
these topics. We would also like to mention that the materials presented in this subsection is not much more than a formalization of what is commonly known as nice tree decomposition~\cite{Kloks94}. We give the details here to make the presentation self-content. 
In the sequel, for a positive integer $q$ we denote by $[q]=\{1,2,\ldots,q\}$. For a function $f$ by $f[x\to y]$ we denote a function defined as follows:
\begin{align*}
f[x\to y](z) = \begin{cases} y & \text{if }z=x\text{,} \\ f(z) & \text{otherwise.}\end{cases}
\end{align*}
Note that this definition is correct regardless whether $x$ was in the domain of $f$ or not. If $x$ belongs to the domain of $f$, then by $f[x\to \bot]$ we denote the function $f\setminus \{(x,f(x))\}$, i.e., $f$ with $x$ deleted from the domain.

Let $k$ be a positive integer and let $\Labels=\{1,2,\ldots,k\}$ be an alphabet of $k$ labels. We now define a family $\termfam$ of terms; each term $\term\in \termfam$ will have a prescribed subset $\used(\term)\subseteq \Labels$ of labels {\em{used}} by $\term$, and a graph $\baggraph(\term)$ with vertex set $\used(\term)$.
\begin{itemize}
\item We have a {\em{leaf term}} $\leaf\in \termfam$, with $\used(\leaf)=\emptyset$ and $\baggraph(\leaf)$ being the empty graph.
\item If $\term\in \termfam$ and $i\in \Labels\setminus \used(\term)$, then we can create an {\em{introduce term}} $\introduce_i(\term)\in \termfam$, with $\used(\introduce_i(\term))=\used(\term)\cup \{i\}$ and $\baggraph(\introduce_i(\term))$ being $\baggraph(\term)$ with an isolated vertex $i$ introduced.
\item If $\term\in \termfam$ and $i\in \used(\term)$, then we can create a {\em{forget term}} $\forget_i(\term)\in \termfam$, with $\used(\forget_i(\term))=\used(\term)\setminus \{i\}$ and $\baggraph(\forget_i(\term))=\baggraph(\term)\setminus \{i\}$.
\item If $\term\in \termfam$, $i,j\in \used(\term)$, $i\neq j$ and $ij\notin E(\baggraph(\term))$, then we can create an {\em{introduce edge term}} $\edge_{i,j}(\term)\in \termfam$, with $\used(\edge_{i,j}(\term))=\used(\term)$ and $\baggraph(\edge_{i,j}(\term))$ being $\baggraph(\term)$ with edge $ij$ added. 
\item Let $q\geq 2$ be any integer. Suppose that there are terms $\term_1,\term_2,\ldots,\term_q\in \termfam$ such that 
\begin{itemize}
\item $\used(\term_1)=\used(\term_2)=\ldots=\used(\term_q)$, and
\item all the graphs $\baggraph(\term_1),\baggraph(\term_2),\ldots,\baggraph(\term_q)$ are edgeless.
\end{itemize}
Then we can create a {\em{join term}} $\join(\term_1,\term_2,\ldots,\term_q)\in \termfam$, with 
$$\used(\join(\term_1,\term_2,\ldots,\term_q))=\used(\term_1)=\used(\term_2)=\ldots=\used(\term_q),$$
and $\baggraph(\join(\term_1,\term_2,\ldots,\term_q))$ being the edgeless graph on vertex set $\used(\join(\term_1,\term_2,\ldots,\term_q))$.
\end{itemize}
The family $\termfam$ comprises all the terms that can be built from leaf terms using introduce, forget, introduce edge, and join terms. Note that the join terms can have arbitrarily large arity, but has to be at least $2$. We define the {\em{length}} of the term $\term$, denoted $|\term|$, as the total number of operators $\leaf,\introduce_i,\forget_i,\edge_{i,j},\join$ used in it.

Terms from $\termfam$ have a natural interpretation as expressions building graphs with at most $k$ distinguished vertices. More formally, with every term $\term\in \termfam$ we associate a pair $\labG[\term]=(G[\term],\lab[\term])$, called a {\em{labelled graph}}, where $G[\term]$ is a graph and $\lab[\term]$ is bijection between some subset of $V(G[\term])$ of cardinality $|\used(\term)|$, and $\used(\term)$. The bijection $\lab[\term]$ is also called the {\em{labelling}}. We maintain the invariant that $\lab[\term]$ is an isomorphism between the graph induced by its domain in $G[\term]$ and the graph $\baggraph(\term)$; this invariant follows by a trivial induction from the definition to follow. The labelled graph $\labG[\cdot]$ is defined as follows:
\begin{itemize}
\item If $\term=\leaf$, then $G[\leaf]$ is an empty graph and $\lab[\leaf]$ is an empty function.
\item If $\term=\introduce_i(\term')$ for some $\term'\in \termfam$ and $i\in \Labels$, then $G[\term]$ is equal to $G[\term']$ with a new independent vertex $v$ introduced, and $\lab[\term]=\lab[\term'][v\to i]$.
\item If $\term=\forget_i(\term')$ for some $\term'\in \termfam$ and $i\in \Labels$, then $G[\term]=G[\term']$ and $\lab[\term]=\lab[\term'][\lab[\term']^{-1}(i)\to \bot]$.
\item If $\term=\edge_{i,j}(\term')$ for some $\term'\in \termfam$ and $i,j\in \Labels$, $i\neq j$, then $G[\term]$ is equal to $G[\term']$ with an edge between $\lab[\term']^{-1}(i)$ and $\lab[\term']^{-1}(j)$ introduced, and $\lab[\term]=\lab[\term']$. Recall that $ij\notin E(\baggraph(\term'))$, so by the induction hypothesis we have that $\lab[\term']^{-1}(i)$ and $\lab[\term']^{-1}(j)$ are not adjacent in $G[\term']$.
\item Suppose $\term=\join(\term_1,\term_2,\ldots,\term_q)$ for some $\term_1,\term_2,\ldots,\term_q\in \termfam$. Then $G[\term]$ is constructed by taking the disjoint union of $G[\term_1],G[\term_2],\ldots,G[\term_q]$, and, for every $i\in \used(\term_1)=\used(\term_2)=\ldots=\used(\term_q)$, identifying all vertices $\{\lab[\term_j]^{-1}(i)\ \colon\ j=1,2,\ldots,q\}$ into one vertex. This identified vertex is assigned label $i$ in the labelling $\lab[\term]$.
\end{itemize}

We now say that $\term$ is a {\em{construction term}} for graph $G$ if $\used(\term)=\emptyset$ and $G[\term]$ is isomorphic to $G$. As the reader probably suspects, construction terms and tree decompositions are tightly related.

\begin{lemma}[\app]\label{lem:tw-ct}
A graph $G$ has treewidth less than $k$ if and only if it admits a construction term that constructs it and uses at most $k$ labels.
\end{lemma}

Let $\Op=\{\introduce_i : i\in \Sigma\}\cup \{\forget_i : i\in \Sigma\}\cup \{\edge_{i,j} : i,j\in \Sigma, i\neq j\}\cup \{\leaf,\join\}$ be the set of operators used in the terms of $\termfam$. Let us introduce an arbitrary linear order $\tleq$ on the elements of $\Op$: for instance first come operators $\introduce_i$, sorted by $i$, then operators $\forget_i$, sorted by $i$, then operators $\edge_{i,j}$, sorted lexicographically by $(i,j)$, and finally operators $\leaf$ and $\join$. Given this order, we may inductively define a linear order $\tleq$ on the terms from $\termfam$ as follows. Let $\term_1,\term_2$ be two terms, and let $o_1,o_2\in \Op$ be the top-most operations used in $\term_1$ and $\term_2$, respectively. Then relation $\tleq$ between $\term_1$ and $\term_2$ is defined inductively based on the definition for terms of smaller depth.
\begin{itemize}
\item If $o_1\neq o_2$, then $\term_1\tl \term_2$ if $o_1\tl o_2$, and $\term_1\tg \term_2$ if $o_1\tg o_2$.
\item If $o_1=o_2=\leaf$, then $\term_1=\term_2$.
\item If $o_1=o_2\notin \{\leaf,\join\}$, then let $\term_1=o(\term_1')$ and $\term_2=o(\term_2')$, where $o=o_1=o_2$. If $\term_1'=\term_2'$ then $\term_1=\term_2$, if $\term_1'\tl \term_2'$ then 
$\term_1\tl \term_2$, and if $\term_1'\tg \term_2'$ then $\term_1\tg \term_2$.
\item Suppose $o_1=o_2=\join$, and let the arity of the join operation in $\term_1,\term_2$ be equal to $q_1,q_2$, respectively. Let $\term_1=\join(\term_{1,1},\term_{1,2},\ldots,\term_{1,q_1})$ and $\term_2=\join(\term_{2,1},\term_{2,2},\ldots,\term_{2,q_2})$. Since terms $\term_{1,j}$ and $\term_{2,j}$ has smaller depth than $\term_1$ and $\term_2$, respectively, the order $\tleq$ is already defined for them. Hence, we may compare sequences $(\term_{1,1},\term_{1,2},\ldots,\term_{1,q_1})$ and $(\term_{2,1},\term_{2,2},\ldots,\term_{2,q_2})$ lexicographically. If $(\term_{1,1},\term_{1,2},\ldots,\term_{1,q_1})\tl (\term_{2,1},\term_{2,2},\ldots,\term_{2,q_2})$ then $\term_1\tl \term_2$, if $(\term_{1,1},\term_{1,2},\ldots,\term_{1,q_1})\tg (\term_{2,1},\term_{2,2},\ldots,\term_{2,q_2})$ then $\term_1\tg \term_2$, and if $(\term_{1,1},\term_{1,2},\ldots,\term_{1,q_1})=(\term_{2,1},\term_{2,2},\ldots,\term_{2,q_2})$ then $\term_1=\term_2$.
\end{itemize}

Note here that two join terms that differ only in the order of arguments are considered different, even though they construct the same labelled graph. The term where the arguments are sorted nondecreasingly is considered the smallest.

\subsection{Constructing a canonical construction term}

\newcommand{\states}{\mathbb{S}}
\newcommand{\cand}{\mathcal{C}}
\newcommand{\br}{\mathtt{break}}

We are finally ready to prove the main result of this paper.

\begin{theorem}[Theorem~\ref{main:canonization}, restated]\label{thm:canonization}
There exists an algorithm that, given a graph $G$ and a positive integer $k$, in time $2^{\Oh(k^5\log k)}\cdot n^{5}$ either correctly concludes that $\tw(G)\geq k$, or outputs an isomorphism-invariant term $\term$ that constructs $G$ and uses at most $(k+1)\bigadh\in \Oh(k^4)$ labels. Moreover, this term has length at most $\Oh(k^4\cdot n)$.
\end{theorem}
\begin{proof}
Let $k'=(k+1)\bigadh$. Firstly, without loss of generality we assume that $G$ is connected. For disconnected graphs we can apply the algorithm to each connected component $G_1,G_2,\ldots,G_p$ separately, obtaining terms $\term_1,\term_2,\ldots,\term_p$, then sort these terms nondecreasingly so that $\term_1\tleq \term_2\tleq \ldots\tleq \term_p$, and output the term $\term:=\join(\term_1,\term_2,\ldots,\term_p)$. Thus, providing that the construction for a connected graph is isomorphism-invariant, then due to the sorting step so is the construction for disconnected graphs.

We now compute the $k$-improved graph $\imp{G}{k}$, using Lemma~\ref{lem:comp-imp}. If computation of this graph revealed that $\tw(G)\geq k$, then we provide a negative answer to the whole algorithm. Now, we apply Theorem~\ref{thm:redsizes} to the graph $\imp{G}{k}$, obtaining an isomorphism-invariant family of candidate bags $\bagfam$. Observe that since the definition of $\imp{G}{k}$ is invariant w.r.t. isomorphisms of $G$, and the definition of $\bagfam$ is invariant w.r.t. isomorphisms of $\imp{G}{k}$, then the family $\bagfam$ is invariant w.r.t. isomorphisms of $G$. 

Assume for a moment that $G$ has a tree decomposition of width less than $k$. Then, by Lemma~\ref{lem:same-decomp}, so does $\imp{G}{k}$. Consequently, by Theorem~\ref{thm:redsizes} we have that $\bagfam$ captures some tree decomposition of $\imp{G}{k}$ that has width at most $k'-1$. Since $\imp{G}{k}$ is a supergraph of $G$, this tree decomposition is also a tree decomposition of $G$. By Lemma~\ref{lem:sensitivitation} and property~(\ref{pr4}) of Theorem~\ref{thm:redsizes}, we can further infer that $\bagfam$ captures some cs-tree decomposition of $G$ of width at most $k'-1$. Let us denote this cs-tree decomposition by $(T,\beta)$.

The plan for the rest of the proof is as follows. We provide a dynamic programming algorithm that exploits the family $\bagfam$ to compute a term $\term$ that constructs $G$. From the algorithm it will be clear that the definition of $\term$ is isomorphism-invariant. It is possible that the computation of $\term$ fails, but only if $\tw(G)\geq k$: using the captured cs-tree decomposition $(T,\beta)$, we will argue that the algorithm computes some feasible construction term, providing that $\tw(G)<k$. Hence, in case of failure we can safely report that $\tw(G)\geq k$.

Let us define the family of states $\states$ as the family of all the triples $(B,\lab,Z)$, where
\begin{itemize}
\item $B\in \bagfam$;
\item $\lab$ is an injective function from $B$ to $[k']$;
\item $Z=\emptyset$ or $Z$ is the vertex set of a connected component of $G\setminus B$.
\end{itemize}
Observe that $|\states|\leq |\bagfam|\cdot k'!\cdot (n+1)=2^{\Oh(k^5\log k)}\cdot n^3$. For every state $I=(B,\lab,Z)\in \states$, we compute a term $\term[I]$ that constructs the labeled graph $\labG[I]:=(G[B\cup Z]\setminus \binom{B}{2},\lab)$, i.e., the graph $G[B\cup Z]$ with all the edges inside $B$ cleared, and with labelling $\lab$ on $B$. The definition of $\term[I]$ will be invariant w.r.t. isomorphisms of the structure $\labG[I]$. Computation of $\term[I]$ can possibly fail, in which case we denote it by $\term[I]=\bot$. The output term $\term$ is simply defined as $\term[\emptyset,\emptyset,V(G)]$, and using the captured cs-tree decomposition $(T,\beta)$ we will make sure that $\term[\emptyset,\emptyset,V(G)]\neq \bot$ in case $\tw(G)<k$.

To make sure that the inductive definition of $\term[I]$ is well-defined, we also define the {\em{potential}} $\pot$ of a state $I=(B,\lab,Z)$ similar to the one defined  in Theorem~\ref{thm:no-clique-seps}: $\pot(B,\lab,Z)=2|Z|+|B|$. The definition of $\term[I]$ depends only on the terms for states with a strictly smaller potential. Since the potential is always nonnegative, the definition is valid.

Before we proceed to the definition of $\term[I]$, let us introduce one more helpful definition. We often run into situations where we would like to compute the canonical term for a triple $(B,\lab,Z)$ that is not necessarily a state according to our definition, because $Z$ consists of several connected components of $G\setminus B$ rather than at most one. To cope with such situations, we define operator $\br[B,\lab,Z]$. Formally, operator $\br[B,\lab,Z]$ can be applied to triples $(B,\lambda,Z)$ where $B\in \bagfam$, $\lab$ is an injective function from $B$ to $[k']$, and $Z$ comprises vertex sets of some (possibly zero) connected components of $G\setminus B$. The behaviour of $\br[B,\lab,Z]$ is defined as follows:
\begin{itemize}
\item If $Z=\emptyset$ or $G[Z]$ is connected (equivalently, $(B,\lab,Z)\in \states$), then we simply put $\br[B,\lab,Z]=\term[B,\lab,Z]$.
\item If $G[Z]$ consists of more than one connected component, then let $Z_1,Z_2,\ldots,Z_p$ be the vertex sets of these connected components. Let $\term_i=\term[B,\lab,Z_i]$ for $i=1,2,\ldots,p$. If any of the terms $\term_i$ is equal to $\bot$, then we put $\br[B,\lab,Z]=\bot$. Otherwise, by sorting the terms if necessary, assume that $\term_1\tleq \term_2\tleq \ldots \tleq \term_p$. Then $\br[B,\lab,Z]=\join(\term_1,\term_2,\ldots,\term_p)$.
\end{itemize}
Observe that the join operation is valid, since we assumed that term $\term_i$ constructs $(G[B\cup Z_i]\setminus \binom{B}{2},\lab)$, where all the edges between the vertices of $B$ are cleared. We naturally extend the notation $\labG[\cdot]$ to triples that can be arguments of the operator $\br[\cdot]$.

We now proceed to the definition of $\term[I]$ for a state $I=(B,\lab,Z)\in \states$. We generate a family $\cand$ of candidates for $\term[I]$. We put $\term[I]=\bot$ if $\cand=\emptyset$, and otherwise $\term[I]$ is defined as the $\tleq$-minimum element of $\cand$. Elements of $\cand$ reflect possible ways of obtaining the term constructing $\labG[I]$ from simpler terms.

Firstly, if $Z=B=\emptyset$, then we take $\cand=\{\leaf\}$. 

Assume now that $B$ contains some vertex $u$ that is not adjacent to any vertex of $Z$. Then, for every such vertex $u$, we add to $\cand$ the term $\term_{\introduce,u}:=\introduce_{\lab(u)}(\term[I'])$ for $I'=(B\setminus u,\lab[u\to \bot],Z)$. Formally, we add this term only if $I'\in \states$ and $\term[I]\neq \bot$; the same remark holds also for the other elements of $\cand$ to follow. Observe that if $\term[I']$ constructs $\labG[I']$, then $\term_{\introduce,u}$ constructs $\labG[I]$.

Then, for every vertex $v\in Z$ we consider the possibility that $v$ has just been forgotten. Formally, for each $v\in Z$ and each label $i\in [k']\setminus \lab(B)$, we add to $\cand$ the following term:
\begin{align}\label{eq:forget}
\term_{\forget,v,i}:=\forget_i(\edge_{i,j_1}(\edge_{i,j_2}(\ldots \edge_{i,j_q}(\br[I'])\ldots ))),
\end{align}
where $I'=(B\cup \{v\},\lab[v\to i],Z\setminus v)$ and $j_1 < j_2 <\ldots <j_q$ are labels in $\lab$ of neighbors of $v$ in $B$ in the graph $G$. Again, if $\br[\cdot]$ cannot be applied to $I'$ or if $\br[I']=\bot$, then we do not add this candidate. Observe that if $\br[I']$ constructs $\labG[I']$, then $\term_{\forget,v,i}$ constructs $\labG[I]$.

This concludes the definition of the term $\term[I]$; observe that the definition depends only on the definitions for states with strictly smaller potential, as was promised. It can be easily seen by induction that the definition is invariant with respect to isomorphisms of the structure $(G,\labG[I])$, due to taking the $\tleq$-minimum from an invariant family of candidates.\
The following claim shows that, in the end, we obtain a meaningful term provided that $\tw(G) < k$.

\begin{claim}\label{cl:canon:nonempty}
If $\tw(G) < k$ then $\term[\emptyset,\emptyset,V(G)] \neq \bot$.
\end{claim}
\begin{proof}
We proceed by a bottom-up induction on the decomposition $(T,\beta)$. For any $t\in V(T)$ and any injective labelling $\lab\colon \sigma(t)\to [k']$, define $I_{t,\lambda}:=(\sigma(t),\lab,\alpha(t))$. Observe that $I_{t,\lambda}\in \states$, since $\sigma(t)\subseteq \beta(t)\in \bagfam$ and $\bagfam$ is closed under taking subsets, and $G[\alpha(t)]$ is connected since $(T,\beta)$ is connectivity-sensitive. We prove inductively the following statement:
\begin{equation}\label{eq:canon:s1}
\textrm{For each }t\in V(T)\textrm{ and any injective labelling }\lab\colon \sigma(t)\to [k'],
\textrm{we have }\term[I_{t,\lambda}]\neq \bot.
\end{equation}
Observe that if $r$ is the root of $T$, then $I_{r,\emptyset}=(\emptyset,\emptyset,V(G))$, so the statement \eqref{eq:canon:s1} for $r$ is equivalent to the statement of the claim.

Let $t_1,t_2,\ldots,t_p$ be the children of $t$ in $T$ (possibly $p=0$). It is more convenient to prove an even stronger statement:
\begin{align}\label{eq:canon:s2}
&\textrm{For any }X\textrm{ such that }\sigma(t)\subseteq X\subseteq \beta(t)\textrm{, any labeling }\lab_X\colon X\to [k']\textrm{ that extends }\lab,\nonumber\\
&\textrm{and any }Z\textrm{ that is the vertex set of some connected component of }G[\gamma(t)]\setminus X,\\
&\textrm{it holds that }\term[X,\lab_X,Z]\neq \bot.\nonumber
\end{align}
Observe that, again, $(X,\lab_X,Z) \in \states$ since $X\subseteq \beta(t)$.
Statement \eqref{eq:canon:s1} for $t$, which we are trying to prove, is equivalent to statement \eqref{eq:canon:s2} for $X=\sigma(t)$, $\lab_X=\lab$ and $Z=\alpha(t)$. We prove statement \eqref{eq:canon:s2} for all choices of $X,\lab_X,Z$ by induction with respect to $|Z\cap \beta(t)|$.

For the base of the induction, observe that if $Z\cap \beta(t)=\emptyset$, then $Z=\alpha(t_i)$ for some $i\in \{1,2,\ldots,p\}$, since $(T,\beta)$ is connectivity-sensitive. Moreover, since $Z$ is a connected component of $G[\gamma(t)]\setminus X$, then $X\supseteq N(Z)=\sigma(t_i)$. By induction hypothesis for statement \eqref{eq:canon:s1}, we have that $\term[\sigma(t_i),\lab_X|_{\sigma(t_i)},\alpha(t_i)]\neq \bot$. Then it follows that the term 
$$\introduce_{u_1,\lambda_X(u_1)}(\introduce_{u_2,\lambda_X(u_2)}(\ldots \introduce_{u_c,\lambda_X(u_c)}(\term[\sigma(t_i),\lab_X|_{\sigma(t_i)},\alpha(t_i)])\ldots ))$$
is among the candidates for $\term[X,\lab_X,Z]$, where $(u_1,u_2,\ldots,u_c)$ is an arbitrary enumeration of $X\setminus \sigma(t_i)$. This proves that $\term[X,\lab_X,Z]\neq \bot$.

Consider now the induction step when $Z\cap \beta(t)$ is non-empty. Let $v$ be any vertex of $Z\cap \beta(t)$, and let $i$ be any label from $[k']\setminus \lab_X(X)$; since $X\subsetneq \beta(t)$ and $|\beta(t)|\leq k'$, such a label exists. Observe now that from the induction hypothesis for statement \eqref{eq:canon:s2} it follows that $\br[X\cup \{v\},\lab_X[v\to i],Z\setminus \{v\}]\neq \bot$, since in the definition of $\br[\cdot]$ the set $Z\setminus \{v\}$ can only be partitioned into smaller connected components, each of them with a strictly smaller intersection with $\beta(t)$ than $Z\cap \beta(t)$. Therefore, the term $\term_{\forget,v,i}$ defined as in (\ref{eq:forget}) is among the candidates for the value of $\term[X,\lab_X,Z]$, which proves that $\term[X,\lab_X,Z]\neq \bot$.

This concludes the inductive proof of statement \eqref{eq:canon:s2} for all choices of $X,\lab_X,Z$, which also proves the induction step for statement \eqref{eq:canon:s1} (both for leaf and non-leaf nodes). As explained earlier, statement \eqref{eq:canon:s1} for the root of the decomposition proves that the algorithm is correct, i.e., it outputs some construction term providing that $\tw(G)<k$.
\cqed\end{proof}

\newcommand{\lbnd}{\phi}

We are left with establishing the upper bound on the length of the output term, and analysing the running time of the algorithm. To achieve this goal, we inductively bound the lengths of the terms produced by the algorithm. For a state $I=(B,\lambda,Z)$, define $\lbnd(I)$ as follows:
\begin{align}\label{eq:lbnd}
\lbnd(I)=(k'+2)\cdot \max(2|Z|-1,0)+|B|+2.
\end{align}
Observe that if $|Z|\geq 1$, then 
\begin{align}\label{eq:lbndeasy}
\lbnd(I)\leq (k'+2)\cdot 2|Z|.
\end{align}

\begin{claim}\label{cl:length}
For any $I\in \states$, if $\term[I]\neq \bot$ then $|\term[I]|\leq \lbnd(I)$.
\end{claim}
\begin{proof}
We first verify the claim for states $I$ where $Z=\emptyset$. Then it is easy to see that $\term[I]$ consists of a sequence of introduce terms that introduce consecutive vertices, finished by a leaf term. Therefore $|\term[I]|=|B|+1\leq \lbnd(I)$. In the sequel we assume that $Z\neq \emptyset$, and we proceed by induction with respect to the potential $\pot(I)$.

Assume now that $Z\neq \emptyset$ and that $\term[B,\lambda,Z]=\term_{\introduce,u}=\introduce_{\lab(u)}(\term[I'])$ for some vertex $u\in B$, where $I'=(B\setminus u,\lab[u\to \bot],Z)$. Then, using the induction hypothesis we have that:
$$|\term[I]|=1+|\term[I']|\leq 1+\lbnd(I)=1+(k'+2)\cdot (2|Z|-1)+(|B|-1)+2=\lbnd(I).$$

Finally, assume that $\term[B,\lambda,Z]=\term_{\forget,v,i}$, where $\term_{\forget,v,i}$ is defined as in (\ref{eq:forget}). Then we have
$$|\term[B,\lambda,Z]|\leq 1+|B|+|\br[B\cup \{v\},\lab[v\to i],Z\setminus \{v\}]|.$$
Let $Z_1,Z_2,\ldots,Z_p$ be the vertex sets of the connected components of $G[Z\setminus \{v\}]$ (possibly $p=0$), and let $I_j=(B\cup \{v\},\lambda[v\to i],Z_j)$ for $j=1,2,\ldots,p$.

Firstly, we consider the case when $p=0$, or equivalently $Z=\{v\}$. Using our observations about the case $Z=\emptyset$ we infer that $|\br[B\cup \{v\},\lab[v\to i],Z\setminus \{v\}]|=|B|+2$. Then
$$|\term[B,\lambda,Z]|\leq 2|B|+3\leq (k'+2)+|B|+1<\lbnd(I).$$

Secondly, we consider the case when $p>0$. Using inequality (\ref{eq:lbndeasy}) and the fact that $|Z_j|\geq 1$ for each $j=1,2\ldots,p$, we infer that
$$|\br[B\cup \{v\},\lab[v\to i],Z\setminus \{v\}]|\leq 1+\sum_{i=1}^p \lbnd(I_p)\leq 1+\sum_{i=1}^p (k'+2)\cdot 2|Z_i|\leq 1+(k'+2)\cdot 2(|Z|-1).$$
Therefore, we obtain that
$$|\term[B,\lambda,Z]|\leq 1+|B|+1+(k'+2)\cdot 2(|Z|-1)\leq \lbnd(I),$$
which concludes the proof of the claim.
\cqed\end{proof}

Claim~\ref{cl:length} implies that each term $\term[I]$ computed by the algorithm has length at most $\Oh(k'n)$, which in particular proves the claimed upper bound on the length of the output term. For the analysis of the running time of the algorithm, observe that for each state $I\in \states$ we consider $\Oh(k'n)$ possible candidates for $\term[I]$. Each of these candidates is constructed in $\Oh(kn)$ time, since we possibly need to partition $G[Z\setminus \{v\}]$ into connected components. Moreover, each of these candidates has length at most $\Oh(k'n)$, by Claim~\ref{cl:length}. It follows that the $\tleq$-minimum among these candidates can be selected in $\Oh((k'n)^2)$ time. Since $|\states|=2^{\Oh(k^5\log k)}\cdot n^3$, we conclude that the whole algorithm works in time $2^{\Oh(k^5\log k)}\cdot n^5$.
\end{proof}

\subsection{Corollaries of Theorem~\ref{main:canonization}}\label{sec:canon-cors}

We now show how Theorems~\ref{main:isomorphism} and~\ref{simple:canonization} follow directly from Theorem~\ref{main:canonization}.

\begin{proof}[Proof of Theorem~\ref{main:isomorphism}]
We run the algorithm of Theorem~\ref{thm:canonization} on both $G_1$ and $G_2$. If for any of them the algorithm concluded that the treewidth is at least $k$, then we output the appropriate answer. Otherwise, the algorithm returned two terms $\term_1,\term_2$ that construct $G_1,G_2$, respectively. Since $\term_1,\term_2$ are isomorphism-invariant, to verify whether $G_1$ and $G_2$ are isomorphic it suffices to check whether $\term_1=\term_2$.
\end{proof}

\begin{proof}[Proof of Theorem~\ref{simple:canonization}]
Given a graph $G$, we compute the canonical term $\term$ constructing $G$, 
construct an ordering $\phi: V(G) \to [n]$ of the vertices of $G$
according to pre-order in term $\term$ of operations when they become forgotten,
and output the graph $G[\term]$ on the vertex set $[n]$ together with the mapping $\phi$.
If the computation of $\term$ returned that the treewidth of $G$ is at least $k$, then we return the same answer.
\end{proof}

\section{Conclusions}\label{sec:conc}
In this paper we have developed the first fixed-parameter tractable algorithm for \GI parameterized by treewidth. The obvious open question is to improve the running time of our algorithm. 

In this work we focused on keeping the presentation as clear as possible, while targeting at a $2^{\text{poly}(k)}$ FPT algorithm at the same time --- but without any attempt of optimizing the $\textrm{poly}(k)$ factor in the exponent, nor the polynomial factor in $n$.
Although it is reasonable to suspect that the polynomial factor in $n$ in the running time of our algorithm can be reduced to $n^4$, or even $n^3$, by a more careful analysis,
we consider such an improvement of minor importance, and the more challenging question would be to make the whole algorithm run in quadratic, or even linear time. Recall that isomorphism of trees can be verified in linear time~\cite{AhoHU74}, so there is no reason why such a running time should not be achieved also for graphs of bounded treewidth. 

A possible route to improving the polynomial factor could be the alternative approach proposed by Otachi and Schweitzer~\cite{pascale}. In essence, Otachi and Schweitzer show that once an isomorphism-invariant family of potential bags of size $f(k)\cdot n^c$ is constructed, then an FPT isomorphism test can be performed using a variant of the {\em{Weisfeiler-Lehman algorithm}}, which thus can serve as an alternative to our dynamic programming procedure of Section~\ref{sec:canon}. Since the Weisfeiler-Lehman algorithm is very simple, it is possible that the combination of our enumeration algorithm and the techniques of Otachi and Schweitzer can lead to improving the polynomial factor.

For the parametric dependence, we also believe that the factor $2^{\Oh(k^5\log k)}$ is suboptimal. A challenging question would be to improve it to $2^{\Oh(k\log k)}$ or even $2^{\Oh(k)}$. In the current approach, the most significant reason for such a high polynomial in the exponent is the way we handle small sets $S$ in the proof of Lemma~\ref{lem:local-decomposition}.

It is also interesting to investigate whether the results of our work can be used to prove canonical or almost canonical variants of other graph decompositions. Actually, many structural theorems for graphs follow the general approach proposed by Robertson and Seymour in their approximation algorithm for treewidth~\cite{gm13}. In particular, the step of breaking the top adhesion $S$ using a small separator has been used multiple times in various works. Since our work provides a canonical way of performing this step (encompassed in Lemmas~\ref{lem:magical} and Lemma~\ref{lem:local-decomposition}), it might serve as a solid foundation for making other graph decompositions canonical. For concrete decomposition theorems where we hope that our approach could be applicable, let us name (a) the $H$-minor-free structural theorem of Robertson and Seymour~\cite{gm16}, (b) the $H$-topological-minor-free structural theorem of Grohe and Marx~\cite{marx-grohe-arxiv,marx-grohe}, and (c) the decomposition theorem with unbreakable parts given by a superset of the current authors~\cite{CyganLPPS13}.

\paragraph*{Acknowledgements.} We are grateful to Yota Otachi and Pascal Schweitzer for sharing with us their manuscript~\cite{pascale} and helpful comments on our work.
Furthermore, we thank an anonymous reviewer for extensive comments.

\bibliographystyle{abbrv}
\bibliography{tw-iso}

\begin{thebibliography}{10}

\bibitem{AhoHU74}
A.~V. Aho, J.~E. Hopcroft, and J.~D. Ullman.
\newblock {\em The Design and Analysis of Computer Algorithms}.
\newblock Addison-Wesley, 1974.

\bibitem{ArnborgP92}
S.~Arnborg and A.~Proskurowski.
\newblock Canonical representations of partial 2- and 3-trees.
\newblock {\em BIT}, 32(2):197--214, 1992.

\bibitem{BabaiL83}
L.~Babai and E.~M. Luks.
\newblock Canonical labeling of graphs.
\newblock In {\em STOC}, pages 171--183, 1983.

\bibitem{BerryPS10}
A.~Berry, R.~Pogorelcnik, and G.~Simonet.
\newblock An introduction to clique minimal separator decomposition.
\newblock {\em Algorithms}, 3(2):197--215, 2010.

\bibitem{Bodlaender90}
H.~L. Bodlaender.
\newblock Polynomial algorithms for {G}raph {I}somorphism and {C}hromatic
  {I}ndex on partial $k$-trees.
\newblock {\em J. Algorithms}, 11(4):631--643, 1990.

\bibitem{Bodlaender96}
H.~L. Bodlaender.
\newblock A linear-time algorithm for finding tree-decompositions of small
  treewidth.
\newblock {\em SIAM J. Comput.}, 25(6):1305--1317, 1996.

\bibitem{Bodlaender98}
H.~L. Bodlaender.
\newblock A partial $k$-arboretum of graphs with bounded treewidth.
\newblock {\em Theor. Comput. Sci.}, 209(1-2):1--45, 1998.

\bibitem{Bodlaender03}
H.~L. Bodlaender.
\newblock Necessary edges in k-chordalisations of graphs.
\newblock {\em J. Comb. Optim.}, 7(3):283--290, 2003.

\bibitem{open-iwpec08}
H.~L. Bodlaender, E.~D. Demaine, M.~R. Fellows, J.~Guo, D.~Hermelin,
  D.~Lokshtanov, M.~M\"uller, V.~Raman, J.~M.~M. van Rooij, and F.~A. Rosamond.
\newblock Open problems in parameterized and exact computation --- {IWPEC}
  2008.
\newblock {\em Technical Report UU-CS-2008-017, Department of Information and
  Computing Sciences, Utrecht University}, 2008.

\bibitem{BoulandDK12}
A.~Bouland, A.~Dawar, and E.~Kopczy\'nski.
\newblock On tractable parameterizations of {G}raph {I}somorphism.
\newblock In {\em IPEC}, pages 218--230, 2012.

\bibitem{ClautiauxCMN03}
F.~Clautiaux, J.~Carlier, A.~Moukrim, and S.~N{\`{e}}gre.
\newblock New lower and upper bounds for graph treewidth.
\newblock In {\em Experimental and Efficient Algorithms, Second International
  Workshop, {WEA} 2003, Ascona, Switzerland, May 26-28, 2003, Proceedings},
  volume 2647 of {\em Lecture Notes in Computer Science}, pages 70--80.
  Springer, 2003.

\bibitem{0030804}
B.~Courcelle and J.~Engelfriet.
\newblock {\em Graph Structure and Monadic Second-Order Logic --- {A}
  Language-Theoretic Approach}, volume 138 of {\em Encyclopedia of mathematics
  and its applications}.
\newblock Cambridge University Press, 2012.

\bibitem{CyganLPPS13}
M.~Cygan, D.~Lokshtanov, M.~Pilipczuk, M.~Pilipczuk, and S.~Saurabh.
\newblock Minimum bisection is fixed parameter tractable.
\newblock In {\em STOC}, pages 323--332, 2014.

\bibitem{diestel}
R.~Diestel.
\newblock {\em Graph Theory}.
\newblock Springer, 2005.

\bibitem{DowneyF13}
R.~G. Downey and M.~R. Fellows.
\newblock {\em Fundamentals of {P}arameterized {C}omplexity}.
\newblock Texts in Computer Science. Springer, 2013.

\bibitem{FilottiM80}
I.~S. Filotti and J.~N. Mayer.
\newblock A polynomial-time algorithm for determining the isomorphism of graphs
  of fixed genus (working paper).
\newblock In {\em STOC}, pages 236--243, 1980.

\bibitem{FlumGroheBook}
J.~Flum and M.~Grohe.
\newblock {\em Parameterized Complexity Theory}.
\newblock Texts in Theoretical Computer Science. An EATCS Series.
  Springer-Verlag, Berlin, 2006.

\bibitem{marx-grohe-arxiv}
M.~Grohe and D.~Marx.
\newblock Structure theorem and isomorphism test for graphs with excluded
  topological subgraphs.
\newblock {\em CoRR}, abs/1111.1109, 2011.

\bibitem{marx-grohe}
M.~Grohe and D.~Marx.
\newblock Structure theorem and isomorphism test for graphs with excluded
  topological subgraphs.
\newblock In {\em STOC}, pages 173--192, 2012.

\bibitem{shonan}
G.~Z. Gutin, K.~Iwama, and D.~M. Thilikos.
\newblock Parameterized complexity and the understanding, design, and analysis
  of heuristics.
\newblock {\em NII Shonan Meeting Report}, 2013-2, 2013.

\bibitem{HopcroftT72}
J.~E. Hopcroft and R.~E. Tarjan.
\newblock Isomorphism of planar graphs.
\newblock In {\em Complexity of Computer Computations}, pages 131--152, 1972.

\bibitem{HopcroftT73}
J.~E. Hopcroft and R.~E. Tarjan.
\newblock A $v \log v$ algorithm for isomorphism of triconnected planar graphs.
\newblock {\em J. Comput. Syst. Sci.}, 7(3):323--331, 1973.

\bibitem{HopcroftW74}
J.~E. Hopcroft and J.~K. Wong.
\newblock Linear time algorithm for isomorphism of planar graphs (preliminary
  report).
\newblock In {\em STOC}, pages 172--184, 1974.

\bibitem{KawarabayashiM08}
K.~Kawarabayashi and B.~Mohar.
\newblock Graph and map isomorphism and all polyhedral embeddings in linear
  time.
\newblock In {\em STOC}, pages 471--480, 2008.

\bibitem{0015106}
J.~M. Kleinberg and {\'E}.~Tardos.
\newblock {\em Algorithm design}.
\newblock Addison-Wesley, 2006.

\bibitem{Kloks94}
T.~Kloks.
\newblock {\em Treewidth, Computations and Approximations}, volume 842 of {\em
  Lecture Notes in Computer Science}.
\newblock Springer, 1994.

\bibitem{KratschS10}
S.~Kratsch and P.~Schweitzer.
\newblock Isomorphism for graphs of bounded feedback vertex set number.
\newblock In {\em SWAT}, pages 81--92, 2010.

\bibitem{Luks82}
E.~M. Luks.
\newblock Isomorphism of graphs of bounded valence can be tested in polynomial
  time.
\newblock {\em J. Comput. Syst. Sci.}, 25(1):42--65, 1982.

\bibitem{Miller80}
G.~L. Miller.
\newblock Isomorphism testing for graphs of bounded genus.
\newblock In {\em STOC}, pages 225--235, 1980.

\bibitem{Otachi12}
Y.~Otachi.
\newblock Isomorphism for graphs of bounded connected-path-distance-width.
\newblock In {\em ISAAC}, pages 455--464, 2012.

\bibitem{pascale}
Y.~Otachi and P.~Schweitzer.
\newblock Reduction techniques for {G}raph {I}somorphism in the context of
  width parameters.
\newblock {\em CoRR}, abs/1403.7238, 2014.

\bibitem{ponomarenko}
I.~Ponomarenko.
\newblock The isomorphism problem for classes of graphs closed under
  contraction.
\newblock {\em Journal of Soviet Mathematics}, 55(2):1621--1643, 1991.

\bibitem{gm13}
N.~Robertson and P.~D. Seymour.
\newblock Graph {M}inors {XIII}. {T}he {D}isjoint {P}aths problem.
\newblock {\em J. Comb. Theory, Ser. B}, 63(1):65--110, 1995.

\bibitem{gm16}
N.~Robertson and P.~D. Seymour.
\newblock Graph {M}inors {XVI}. {E}xcluding a non-planar graph.
\newblock {\em J. Comb. Theory, Ser. B}, 89(1):43--76, 2003.

\bibitem{Schoning88}
U.~Sch{\"o}ning.
\newblock Graph {I}somorphism is in the low hierarchy.
\newblock {\em J. Comput. Syst. Sci.}, 37(3):312--323, 1988.

\bibitem{Tarjan85}
R.~E. Tarjan.
\newblock Decomposition by clique separators.
\newblock {\em Discrete Mathematics}, 55(2):221--232, 1985.

\bibitem{BevernFGR13}
R.~van Bevern, M.~R. Fellows, S.~Gaspers, and F.~A. Rosamond.
\newblock Myhill-nerode methods for hypergraphs.
\newblock In {\em ISAAC}, volume 8283 of {\em Lecture Notes in Computer
  Science}, pages 372--382, 2013.

\bibitem{Weinberg}
H.~Weinberg.
\newblock A simple and efficient algorithm for determining isomorphism of
  planar triply connected graphs.
\newblock {\em Circuit Theory}, 13:142--148, 1966.

\bibitem{YamazakiBFT99}
K.~Yamazaki, H.~L. Bodlaender, B.~de~Fluiter, and D.~M. Thilikos.
\newblock Isomorphism for graphs of bounded distance width.
\newblock {\em Algorithmica}, 24(2):105--127, 1999.

\end{thebibliography}

\newpage

\appendix
\section*{Appendix}
\begin{proof}[Proof of Lemma~\ref{lem:improval}]
Let $H=\imp{G}{k}$. Take any $x,y\in V(G)$ and assume that $xy\notin E(H)$, so in particular $xy\notin E(G)$. Then by the definition of $H=\imp{G}{k}$, we have that the maximum vertex flow between $x$ and $y$ in $G$ has size less than $k$. By Menger's theorem this means that there exists a separation $(A,B)$ of $G$ that has order less than $k$, and such that $x\in A\setminus B$ and $y\in B\setminus A$. We claim that $(A,B)$ is also a separation of $H$. Indeed, for any pair of vertices $u\in A\setminus B$ and $v\in B\setminus A$, the fact that $|A\cap B|<k$ certifies that $\conn_G(u,v)<k$, which means that $uv\notin E(H)$ by the definition of $H$. Since $x\in A\setminus B$, $y\in B\setminus A$, and $(A,B)$ is a separation of order less than $k$ in $H$, then this certifies that $\conn_H(x,y)<k$. As $x,y$ were chosen arbitrarily, it follows that $H$ is $k$-complemented.
\end{proof}

\begin{proof}[Proof of Lemma~\ref{lem:same-decomp}]
Let $(T,\beta)$ be a tree decomposition of $G$ of width less than $k$.
In order to show that $(T,\beta)$ is also a tree decomposition of $H:=\imp{G}{k}$, it suffices to show that for any $xy\in E(H)\setminus E(G)$ there exists some $t\in T$ such that $x,y\in \beta(t)$.

For the sake of contradiction assume that no such $t$ exists. Let $T^0_x$ and $T^0_y$ be the subtrees of $T$ induced by the nodes whose bags contain $x$ and $y$, respectively.
We know that $T^0_x$ and $T^0_y$ are connected and vertex-disjoint.
Let $T_y$ be the connected component of $T \setminus V(T^0_x)$ that contains $T^0_y$ as a subgraph,
and let $T_x = T \setminus V(T_y)$; note that $T_x$ is connected and contains $T^0_x$ as a subgraph.
Observe that $(V(T_x),V(T_y))$ forms a partition of $V(T)$.
Let $A=\bigcup_{t\in V(T_x)} \beta(t)$ and $B=\bigcup_{t\in V(T_y)} \beta(t)$; note that $x\in A\setminus B$ and $y\in B\setminus A$.
Observe that since $T$ is a tree, there exists only one edge $t_xt_y$ of $T$ that connects a node from $V(T_x)$ with a node from $V(T_y)$.
From the properties of a tree decomposition it follows that $A\cap B=\beta(t_x)\cap \beta(t_y)$.
Moreover, observe that $x \in \beta(t_x) \setminus \beta(t_y)$, thus $|A \cap B| = |\beta(t_x) \cap \beta(t_y)| < |\beta(t_x)| \leq k$.
As every vertex of $G$ is contained in some bag of $(T,\beta)$, it follows that $(A,B)$ is a separation of order less than $k$ that separates $x$ and $y$.
This proves that $\conn_G(x,y)<k$, contradicting the assumption that $xy\in E(H)$.
\end{proof}

\begin{proof}[Proof of Lemma~\ref{lem:comp-imp}]
If $|E(G)|>(k-1)n$, then we can output that $\tw(G)\geq k$, since a graph of treewidth less than $k$ is $(k-1)$-degenerate. Hence assume that $|E(G)|\leq (k-1)n$. We perform a brute-force algorithm that follows immediately from the definition: For every pair of vertices, we compute the maximum flow between them using the Ford-Fulkerson algorithm, stopping the computation if the size of the flow exceeded $k-1$. Thus we run at most $k$ iterations of the Ford-Fulkerson algorithm, and each iteration takes $\Oh(|V(G)|+|E(G)|)=\Oh(kn)$ time. Since we perform this procedure for every pair of vertices, the running time of $\Oh(k^2n^3)$ follows.
\end{proof}

\begin{proof}[Proof of Lemma~\ref{lem:sensitivitation}]
Let $(T,\beta)$ be a tree decomposition of $G$ of width $k$ and adhesion width $\ell$. We prove the following statement by a bottom-up induction on $(T,\beta)$: For every $t\in V(T)$ and every vertex set $Z$ of a connected component of $G[\alpha(t)]$, there exists a tree decomposition $(T_{t,Z},\beta_{t,Z})$ of $G[N[Z]]$ such that:
\begin{enumerate}[(a)]
\item\label{h1} $(T_{t,Z},\beta_{t,Z})$ is connectivity-sensitive, 
\item\label{h2} $(T_{t,Z},\beta_{t,Z})$ has width at most $k$ and adhesion width at most $\ell$, 
\item\label{h3} every bag of $(T_{t,Z},\beta_{t,Z})$ is a subset of some bag of $(T,\beta)$, and
\item\label{h4} $N(Z)$ is contained in the root bag of $(T_{t,Z},\beta_{t,Z})$.
\end{enumerate}
Observe that in this definition it holds that $N(Z)\subseteq \sigma(t)$. Since $G$ is connected, for the final decomposition $(T',\beta')$ we may take $(T_{r,V(G)},\beta_{r,V(G)})$, where $r$ is the root of $(T,\beta)$.

Let us focus on one choice of $t,Z$. Let $t_1,t_2,\ldots,t_p$ be the children of $t$ in $T$ (possibly $p=0$). For $i=1,2,\ldots,p$, let $Z_i=Z\cap \alpha(t_i)$, and let $Z_i^1,Z_i^2,\ldots,Z_i^{q_i}$ be the vertex sets of the connected components of $G[Z_i]$. For $i=1,2,\ldots,p$ and $j=1,2,\ldots,q_i$, let $(T_{t_i,Z_i^j},\beta_{t_i,Z_i^j})$ be the decomposition of $G[N[Z_i^j]]$ that satisfies properties (\ref{h1}), (\ref{h2}), (\ref{h3}), (\ref{h4}); existence of this decomposition is asserted by the induction hypothesis. We construct the decomposition $(T_{t,Z},\beta_{t,Z})$ by creating one bag $\beta(t)\cap N[Z]$, and attaching all the decompositions $(T_{t_i,Z_i^j},\beta_{t_i,Z_i^j})$ below it as subtrees. Let $t_i^j$ be the root node of the attached decomposition $(T_{t_i,Z_i^j},\beta_{t_i,Z_i^j})$. It is straightforward to verify that $(T_{t,Z},\beta_{t,Z})$ is indeed a tree decomposition of $G[N[Z]]$. We now verify that the requested properties are satisfied.
\begin{itemize}
\item For property~(\ref{h1}), the only checks not implied by the induction hypothesis are as follows:
\begin{itemize}
\item We need to verify that $G[N[Z]]$ is connected, but this follows from the fact that $G[Z]$ is connected.
\item We need to verify that $G[\alpha_{t,Z}(t_i^j)]$ is connected and that $N(\alpha_{t,Z}(t_i^j))=\sigma_{t,Z}(t_i^j)$. However, we have that $\alpha_{t,Z}(t_i^j)=Z_i^j$, which induces a connected graph by its definition, and that $\sigma_{t,Z}(t_i^j)=N(Z_i^j)$ by the definition of $(T_{t_i,Z_i^j},\beta_{t_i,Z_i^j})$ and property~(\ref{h4}) for it.
\end{itemize}
\item For properties~(\ref{h2}) and~(\ref{h3}), it suffices to observe that $\beta(t)\cap N[Z]\subseteq \beta(t)$ and that $\sigma_{t,Z}(t_i^j)\subseteq \sigma(t_i)$.
\item Property~(\ref{h4}) follows directly from the definition of $Z$ and of the top bag of $(T_{t,Z},\beta_{t,Z})$.
\end{itemize}
This concludes the step of the induction.
\end{proof}

\begin{proof}[Proof of Lemma~\ref{lem:tw-ct}]
From right to left, let $\term$ be a term that constructs $G$ and uses at most $k$ labels. For each subterm $\term'$ of $\term$, create one node $t_{\term'}$ with associated bag $\beta(t_{\term'})$ being the domain of $\lab[\term']$. Since $\term$ uses at most $k$ labels, it follows that $|\beta(t_{\term'})|\leq k$. Now connect nodes $t_{\term'}$ into a tree using the structure inherited from the term $\term$: for any two subterms $\term'$, $\term''$, connect $t_{\term'}$ and $t_{\term''}$ if and only if $\term''$ appears as an argument of the top-most operation in $\term'$, or vice-versa. Let $T$ be the obtained tree. It is straightforward to verify that the $(T,\beta)$ is a tree decomposition of $G$, and we already verified that it has width less than $k$.

From left to right, let $(T,\beta)$ be a tree decomposition of $G$ of width less than $k$. We apply a bottom-up induction on $(T,\beta)$. More precisely, for every $t\in V(T)$ and every injective labeling $\lab\colon \sigma(t) \to [k]$ we construct a term $\term_{t,\lab}$ that constructs $(G[\gamma(t)]\setminus \binom{\sigma(t)}{2},\lab)$, i.e., the graph $G[\gamma(t)]$ with all the edges inside $\sigma(t)$ cleared, and with labeling $\lab$ on $\sigma(t)$. The final term $\term$ will be just $\term_{r,\emptyset}$, where $r$ is the root of the tree $T$.

Let us take any $t\in V(T)$, and let $t_1,t_2,\ldots,t_p$ be the children of $t$ in $T$ (possibly $p=0$). Let $\lab'$ be any injective extension of $\lab$ onto $\beta(t)$; such an extension exists since $|\beta(t)|\leq k$. We first construct an auxiliary term $\term'$, which will construct $(G[\gamma(t)]\setminus \binom{\beta(t)}{2},\lab')$. The construction of $\term'$ distinguishes two cases: either $p=0$ or $p>0$.

Consider first the case when $p=0$, i.e., $t$ is a leaf node. Then we can take 
\begin{align*}
\term' = \introduce_{\lab'(u_1)}(\introduce_{\lab'(u_2)}(\ldots \introduce_{\lab'(u_{|\beta(t)|})}(\leaf)\ldots )),
\end{align*}
where $(u_1,u_2,\ldots,u_{|\beta(t)|})$ is an arbitrary ordering of the vertices of $\beta(t)$.

Consider now the case when $p>0$. For $i=1,2,\ldots,p$, let $\term_i$ be equal to $\term_{t_i,\lab_i}$, where $\lab_i=\lab'|_{\sigma(t_i)}$. Existence of $\term_i$ is asserted by the induction hypothesis for the node $t_i$. Construct $\term_i'$ from $\term_i$ by applying $\introduce$ operation for all the labels that are used in $\lab'$, but not in $\lab_i$. Then we can take
\begin{align*}
\term' := \begin{cases} \term_1' & \text{if }p=1{,} \\ \join(\term_1',\term_2',\ldots,\term_p') & \text{if }p>1{.}\end{cases}
\end{align*}
It is straightforward to see that in both cases $\term'$ constructs $(G[\gamma(t)]\setminus \binom{\beta(t)}{2},\lab')$, as claimed.

Now we need to show how to obtain $\term_{t,\lab}$ from $\term'$. Let $L$ be the set of labels used in $\lab'$ and let $L_\sigma\subseteq L$ be the set of labels used in $\lab$. To obtain $\term_{t,\lab}$ from $\term'$, we perform the following two operations:
\begin{itemize}
\item Apply $\edge$ operations to all the pairs of labels $\{j_1,j_2\}\in \binom{L}{2}\setminus \binom{L_\sigma}{2}$ such that $\lab'^{-1}(j_1)\lab'^{-1}(j_2)\in E(G)$, in any order.
\item Apply $\forget$ operations to all the labels of $L\setminus L_\sigma$, in any order.
\end{itemize}
It is straightforward to see that $\term_{t,\lab}$ constructs $(G[\gamma(t)]\setminus \binom{\sigma(t)}{2},\lab)$, as claimed. This concludes the inductive proof.
\end{proof}

\end{document}